\documentclass[a4paper,11pt]{article}
\usepackage{amsfonts}
\usepackage{amssymb}
\usepackage{amsmath}
\usepackage{amsthm}
\usepackage{verbatim}
\usepackage{algorithmic}
\usepackage[T1]{fontenc}
\usepackage{color}
\theoremstyle{plain}
\newtheorem{theorem}{Theorem}[section]

\theoremstyle{definition}
\newtheorem{definition}[theorem]{Definition}

\begin{document}
\title{On Games and Computation}
\author{Antti Kuusisto\\
{\small University of Helsinki, Tampere University}\\
{\small }}

\date{}

\maketitle

\begin{abstract}
\noindent
We introduce and investigate a range of general notions of a game. Our
principal notion is based on a set of 
agents modifying a relational structure in a discrete evolution sequence. 
We also introduce and study a variety of ways to 
model partial and erroneous information in the setting.
We discuss the connection of the related general setting to logic
and computation formalisms, with emphasis on the recently introduced
Turing-complete logic based on game-theoretic semantics.
\end{abstract}

\section{Introduction}

We introduce and investigate a range general formalisations of
the notion of a \emph{game}. Games here refer to multiplayer interaction 
systems as conceived in, e.g., the field of multiagent systems. Our main 
formalisation is an iterative setting where the players jointly modify a 
relational structure in a discrete sequence of steps.
The approach is quite general, and generality is indeed 
one of our principal aims.

\medskip

\medskip

\noindent
To gain intuition into the setting,
the relational structures can be considered, e.g., to represent the board of
some board game---chess for example---at different points of time.
The individual pawns and other pieces can then be
naturally modeled by constant symbols or singleton predicates,
for example. The players move the
pieces about, i.e., modify the relational structure.

\medskip

\medskip

\noindent
In the general setting, we put no limitations to what the modifications could be like in a
particular scenario. 
It may be possible to remove domain elements and
introduce new ones to the structures.
Likewise, it may be possible to remove tuples from the relations of the structures
and introduce new ones.
Each game round corresponds intuitively to a new, modified structure.
In any particular modeling scenario, only the game rules restrict the set of
allowed modifications in each round. A function modeling \emph{chance} is
also included into the setting to enable investigations
requiring related features and capacities.

\medskip

\medskip

\noindent
Board games, however, are only a starting point. The setting we
define is intended to provide a very general modeling framework.
The framework aims to offer a wide range of options for
studying different kinds of
interaction scenarios involving a \emph{concrete dynamic
environment} (the changing relational structures) and a 
set of \emph{agents} acting in that environment.
This will then be connected to a very general approch to logic
using a powerful, Turing-complete logic formalism introduced
recently in \cite{tc}. The logic provides a full range of ways to 
formally control the new setting.

\medskip

\medskip

\noindent
Using relational structures as the starting point of our 
formal systems has two principal advantages. Firstly,
relational structures are highly general as well as natural,
being able to model more or less everything in a
flexible way. Secondly, relational structures enable us to
\emph{directly use different logics} to control the 
time evolution and flow of changing structures.

\medskip

\medskip

\noindent
Logic plays a crucial role in our study. We first observe that the
Turing-complete logic $\mathcal{L}$ of \cite{tc} is
intimately connected to our main formalisation of the notion of a game.
Indeed, the evaluation of formulae via the game-theoretic semantics of $\mathcal{L}$ is all
about modifying relational strucures, so $\mathcal{L}$ can be viewed as a
particular game system included in our formal framework of games.
Conversely, we briefly analyse ways to directly simulate formal game 
evolutions of our framework within the setting of $\mathcal{L}$.
Moreover, we discuss further general ways to control  game systems via logic,
including, e.g., ways of representing knowledge of agents and beyond.

\medskip

\medskip

\noindent
In addition to obviously considering perfect information scenarios, we
introduce a simple and natural yet highly general way to deal with partial and
potentially false information.
The approach is based on two maps. The \emph{perception map} provides---based on the current 
relational structure---a
\emph{mental model} that reflects the way an agent sees the actual
current world (i.e., the current relational structure) and other relevant
factors concerning the agent.
The agent then acts, in one way or another, using the mental model to 
decide upon the particular course of action.
The chosen actions can depend on the agent's (possibly limited) reasoning capacities.
All this is captured formally by a \emph{decision map} that takes the mental model as an input
and outputs a specific action. The mental model can be a relational structure, but we also 
consider more elaborate approaches to better account for incomplete information issues.

\medskip

\medskip

\noindent
To supplement our principal notion of a system, we also consider some 
generalizations. For example, we consider ways to abstract away the discrete iteration steps
leading from a structure to another. This gives rise to a potentially continuous flow of
structures. Furthermore, the approach
provides a way to model situations with infinite past, cyclic time, et cetera.

\medskip

\medskip

\noindent
There is of course a vast literature investigating
notions related to our study, especially in the field of multiagent
systems \cite{wool}. The concurrent game models
used in Alternating-time temporal logic \cite{alur} relate to
our notion of a game system, but the main focus is not on relational structures there. In first-order temporal
logics (see, e.g., \cite{hodg} and the 
references therein), however, the setting typically involves a flow of
relational structures.
Formalisms that bear some similarity to the original motivations of 
the logic $\mathcal{L}$, as
given in \cite{tc}, include, e.g., Abstract
State Machines \cite{borg}, but that approach is---unlike $\mathcal{L}$---only
remotely related to our study of multiagent interactions.
The idea that the \emph{general notion} of a
game should be formulated in terms of agents 
jointly \emph{modifying a relational structure} (or \emph{model})
has been stated in \cite{a15k, a15k2}
and formulated in further detail in \cite{a18k}. We 
elaborate on those suggestions, developing an
elaborate notion of a game system and drawing links with logic.
This leads to a framework with a reasonably
flexible capacity to model more or less every
possible dynamic interaction scenario.

\medskip

\medskip

\noindent
Our approach is foundational and thus
we provide relatively detailed discussions of most definitions we give, justifying
the theoretical and formal choices. After the brief technical preliminaries in Section \ref{preliminaries}, we
introduce and discuss formal notions of a system (i.e., notions of a game or interaction framework) in Section \ref{systems}.
In Section \ref{systems and logic} we then draw
connections to logic, especially the Turing-complete logic $\mathcal{L}$,
but also other systems. In addition to considering the connections of $\mathcal{L}$ to games, we
also directly analyse some of the more fundamental properties of the logic.

\section{Preliminaries}\label{preliminaries}

\noindent
The power set of a set $S$ is denoted by $\mathcal{P}(S)$.
For any signature $\sigma$, the empty $\sigma$-structure is
in general allowed. Note that the empty sructure is
not the same object as $\emptyset$. We suppose 
this holds holds even if $\sigma = \emptyset$.

\medskip

\medskip

\noindent
A structure (or model) typically refers to a first-order model as conceived in standard logic. However,
below structures can also be more general objects, such as---to name a few of the
many possibilities---sets or classes of first-order structures; sets of logical
formulas; first-order models with relations having probabilistic weights on
the relation tuples; or pairs $(\mathfrak{B},f)$ where $\mathfrak{B}$ is a
first-order model and $f$ and assignment function mapping some set of
variable symbols into the domain of $\mathfrak{B}$. This generality 
can be advantageous. For example, a set of first-order structures can 
represent a  set of conceived possible worlds, while a reasonable setting for 
modeling quantum phenomena could be to consider sets of first-order models, each 
model having a complex number weight.\footnote{In one simple case, the domain of the 
first-order models in that setting would correspond to space coordinates.}
\emph{However, standard relational first-order 
models are by far the most important notion of a
structure that we consider below, providing background intuitions for most of the
discussed technical as
well as conceptual issues.} However, we
use the word \emph{model} as a synonym for \emph{structure}, and refer to \emph{first-order models}
when it is indeed only standard first-order models that we are considering.

\medskip

\medskip

\noindent
We assume that each structure can be associated with a signature $\sigma$ that 
relates to the objects of that structure. In the paradigmatic case of standard
first-order models, the signature is as defined in standard logic.
We define \emph{relational} first-order models to have a purely relational 
signature, so constant symbols and obviously function symbols are not included.
First-order models are not assumed to be finite by default, as is sometimes the case in
mathematics relating to computation (especially finite model theory).

\section{Systems}\label{systems}

In this section we define a general notion of a system. We 
begin with some preliminary definitions.

\medskip

\medskip

\noindent
Consider a triple $(\sigma, A, I)$, where $\sigma$ is a
signature, $A$ a set of actions and $I$ a set of agents (or agent names).
Let $S$ be a set of $\sigma$-structures.
An $(S,A,I)$-\emph{sequence} is a finite sequence $$(\mathfrak{B}_0,\textbf{a}_0,\mathfrak{B}_1,
\textbf{a}_1,\dots , \mathfrak{B}_{k},\textbf{a}_{k})$$ where $\mathfrak{B}_i\in S$ 
and $\mathbf{a}_i\in A^I$ for each $i\leq k$. We note that also the empty sequence, denoted by $\epsilon$, is
considered an $(S,A,I)$-sequence.

\medskip

\medskip

\begin{definition}\label{firstsystem}
\normalfont
A \emph{system frame base} over $(\sigma,A,I)$ is a pair $(S,F)$ 
such that the following conditions hold:
\begin{enumerate}
\item
$S$ is a set of $\sigma$-structures.
%
%
%
%
\item
$F$ is a function $F:\ T \rightarrow\
\mathcal{P}(S)$, where $T$ is some subset of
the set of all $(S,A,I)$-sequences.\hfill\qedsymbol
\end{enumerate}
\end{definition}

\medskip

\medskip

\noindent
Intuitively, a system frame base consists of a set $S$ of possible worlds
and a function $F$ that (nondeterministically) indicates how finite sequences of possible worlds are allowed to
evolve to longer sequences. The sequences correspond to time evolutions of possible worlds.


\medskip

\medskip

\noindent
In a bit more detail, consider a sequence
$$(\mathfrak{B}_0,\textbf{a}_0,\mathfrak{B}_1,
\textbf{a}_1,\dots , \mathfrak{B}_{k-1},\textbf{a}_{k-1},\mathfrak{B}_{k})$$
of possible worlds $\mathfrak{B}_i$ and (tuples of) actions $\mathbf{a}_i\in A^I$ 
carried out\footnote{The actions in $\textbf{a}_i$ can most naturally be
considered to be carried simultaneously in $\mathfrak{B}_i$. However, interpreting these actions
simultaneous is by no means the only possibility.} in those possible worlds. This sequence ends with the possible 
world $\mathfrak{B}_{k}$ that could be considered the \emph{current possible world}, or the 
current state of affairs. Now, if the tuple of actions $\mathbf{a}_k\in A^I$ is carried out in
the current possible world $\mathfrak{B}_k$, we 
get the extended sequence
$$(\mathfrak{B}_0,\textbf{a}_0,\mathfrak{B}_1,
\textbf{a}_1,\dots , \mathfrak{B}_{k}, \textbf{a}_k).$$
Now the function $F$ gives the set
$$F\bigl((\mathfrak{B}_0,\textbf{a}_0,\mathfrak{B}_1,
\textbf{a}_1,\dots , \mathfrak{B}_{k}, \textbf{a}_k)\bigr)$$
of new possible worlds, and one these worlds will ultimately become the
new current possible world. Note indeed that $F$ does not deterministically give a single
new current possible world, but instead only a set of new candidates.
In the special case where $F$ outputs the empty set, it is
natural to interpret the situation so that
the actions $\textbf{a}_k$ lead to termination of the evolution. 

\medskip

\medskip

\noindent
Note also that the domain of the function $F$ is specified to be a
subset $T$ of the set of all $(S,A,I)$-sequences, with no particular restrictions on $T$.
Thus it can happen that $F$ is defined even on some $(S,A,I)$-sequences that do not
belong to the set $T_F$ of all possible sequences that $F$ gives rise to.\footnote{The set $T_F$ is 
the set of sequences obtained by starting from the empty sequence $\epsilon$ and inductively
generating all possible sequences according to what $F$ outputs.} This feature could of course be
avoided by putting extra conditions on $F$. But, this extra flexibility and generality in the
definition of $F$ can also be beneficial.\footnote{For example, we could define some function $F_r$ 
according to some natural behaviour restriction $r$ and then study what 
kinds of evolutions the function $F_r$ would allow when starting from a sequence $t\not=\emptyset$
such that $t\not\in T_{F_r}$.}

\medskip

\medskip

\noindent
Since $F$ is indeed a partial function on the set of $(S,A,I)$-sequences, there indeed may be cases
where $F$ gives no output. This is subtly different from the case where $F$ outputs
the empty set. Supposing $F$ is
undefined on the input $t  = (\mathfrak{B}_0,\textbf{a}_0,\mathfrak{B}_1,
\textbf{a}_1,\dots , \mathfrak{B}_{k}, \textbf{a}_k)$, we can interpret this to mean, e.g.,
that the tuple $\textbf{a}_k$ contains some \emph{forbidden} actions in 
the possible world $\mathfrak{B}_{k}$ when the history leading to $\mathfrak{B}_k$
is $(\mathfrak{B}_0,\textbf{a}_0,\mathfrak{B}_1,
\textbf{a}_1,\dots , \mathfrak{B}_{k-1}, \textbf{a}_{k-1})$.
If an evolution terminates this way
due a tuple of actions that is not allowed, the situation is indeed subtly different from
termination resulting in from $F$ outputting $\emptyset$ (which corresponds to 
termination via an allowed tuple of actions). Of course---in different scenarios---one could
talk about \emph{possible} or \emph{available} actions rather than allowed and forbidden actions. It all
depends on the background interpretations.

\medskip

\medskip

\noindent
It is often natural to allow \emph{non-actions} in addition to actions.
Then we can define $A$ so that it contains a special symbol (or
perhaps many special symbols) that correspond to taking no action
whatsoever. For example, suppose $A = \{x,y\}$ with $x$ indicating no action
taken and $y$ corresponding to some action. Let $I = \{0,1\}$. Then
the tuples $(x,x),(x,y)$ and $(y,x)$ correspond to situations with non-actions.
If $F$ is undefined, say, on some sequence ending with $(x,x)$, then this can correspond 
for example to a scenario where at least one action in the action tuple is required and
the total non-action tuple $(x,x)$ is simply not allowed or somehow impossible.

%

%

\medskip

\medskip

\noindent
Now, $F$ is indeed nondeterministic in the sense that it only gives a \emph{set} of new 
possible worlds in a frame base $(S,F)$.
Therefore, to decide which one of the new possible
worlds given by $F$ becomes the new current possible world,
we define the notion of a system frame. The key is simply to define a choice function $G$
that picks a new possible world from the set of possibilities given by $F$.

\medskip

\medskip

\begin{definition}
A \emph{system frame} over $(\sigma,A,I)$ is triple $(S,F,G)$ 
such that the following conditions hold:
\begin{enumerate}
\item
$(S,F)$ is a system frame base as defined above.
\item
$G:E\  \rightarrow\  S\cup\{ \mathbf{end} \}$ is a 
function with $E \subseteq T\times \mathcal{P}(S)$ where $T$ is the
set of all $(S,A,I)$-sequences. For all inputs $(t,W)$ where $G$ is defined 
and $G((t,W)) \not=\mathbf{end}$, we require that $G((t,W)) \in W$.
%
%
%
%
%
%
%
\end{enumerate}
\end{definition}

\medskip

\medskip

\noindent
Intuitively, $G$ simply chooses one option from the set $W$ of possible worlds given by $F$, and this choice
depends also on the history $t\in T$. When $G$ outputs $\textbf{end}$, the interpretation
can be that $G$ terminates the evolution of the underlying system. When $G$ is undefined, we
can interpret this for example to indicate that $G$ has no resources to determine the output.
Note also that $G$ is undefined or outputs $\textbf{end}$
always when $F$ outputs $\emptyset$.
This reflects the idea that if evolution is terminated due to $F$, then $G$ complies with this and the
evolution indeed will not continue.

\medskip

\medskip

\noindent
The background intuitions between $F$ and $G$ are different; while $F$ provides a set of 
restrictions on how a system could potentially evolve, $G$ determines, within those restrictions,
how the system then actually evolves. Thus $F$ can be seen as providing the rules how a system 
must evolve, and $G$ is a bit like, e.g., \emph{luck} or \emph{chance} that then 
determines what happens within the allowed constraints. More on the interpretation of $F$ and $G$ (and
beyond) will be given later on.

\medskip

\medskip

\noindent
We are now ready to define the notion of a system. To this end, we
first define that a \emph{structure-ended $(S,A,I)$-sequence} is any
sequence that can be obtained by 
extending an $(S,A,I)$-sequence by some structure in $S$. More formally, a
structure ended $(S,A,I)$-sequence is a  
sequence $$(\mathfrak{B}_0,\textbf{a}_0,\dots , \mathfrak{B}_{k-1},
\textbf{a}_{k-1}, \mathfrak{B}_{k})$$
where $(\mathfrak{B}_0,\textbf{a}_0,\dots , \mathfrak{B}_{k-1},
\textbf{a}_{k-1})$ is an $(S,A,I)$-sequence and $\mathfrak{B}_{k}\in S$ with $k\geq 0$.
We then define the notion of a system. This amounts to adding
\emph{agents} $f_i$ that act (choose actions in $A$) in each current possible world.

\medskip

\medskip

\noindent
\begin{definition}
A \emph{system} over $(\sigma,A,I)$ is a structure $(S,F,G,(f_i)_{i\in I})$ 
defined as follows.
\begin{enumerate}
\item
$(S,F,G)$ is a system
frame as defined above.
\item
Every $f_i$ is a function $f_i:V_i
\rightarrow A$ where $V_i$ is a subset of the 
set of all structure-ended $(S,A,I)$-sequences.\ \ \ \hfill\qedsymbol
\end{enumerate}
\end{definition}

\medskip

\medskip

\noindent
Agents are partial functions on the set of
structure-ended $(S,A,I)$-sequences.
Intuitively, an agent makes choices in models of $S$ based on the
current model $\mathfrak{B}_{k}$
and also the $(S,A,I)$-sequence that gave rise to that model.
If an agent is undefined on some entry, this
can perhaps most naturally be interpreted so that the entry is irrelevant for the underlying study,\footnote{The 
same interpretation for the cases where $F$ or $G$ is undefined is also important. Indeed, one reason
for allowing $F$, $G$ and each $f_i$ to be partial functions is to enable finite (or 
otherwise limited) systems to be defined.} to give one option. If an
agent $f_i$ gets permanently removed from the system for some reason, then it can still be technically
desirable to keep $f_i$ defined on sequences that
extend further to the future in order to enable longer and longer evolutions to be
free of entries where functions have no defined value.\footnote{Some crucial action
tuple $\textbf{a}_i = (f_j( t_i ))_{j\in I}$ (where $t_i$ is a structure-ended $(S,A,I)$-sequence)
can then have all its entries defined even if some agents $j\in I$ are not present in
the last world of $t_i$.}
The removed agent can, for example, systematically output some special non-action
symbol (say, $d\in A$). Similar considerations
can concern agents that have not yet entered the system, or have temporarily left the system.
These can be associated with different symbols (say, $u\in A$ and $\mathit{t}\in A$). An
agent who is present, but chooses not to act, would output some
yet another non-action symbol. Using special
outputs for non-actions has the benefit that we can
indeed distinguish reasons why the agent is inactive.

\medskip

\medskip

\noindent
It is at this stage quite clear that together with $G$, the
agents $f_j$ make systems evolve within the constraints 
given by $F$. The agents act in a possible world $\mathfrak{B}_i$, and then $F$
determines, based on the actions, a set $W$ of potential new possible worlds. 
The actual new possible world is then chosen from $W$ by $G$.

\medskip

\medskip

%
%
%
\noindent
The set of \emph{finite evolutions} of a system $(S,F,G,(f_i)_{i\in I})$ is the 
set that contains all structure-ended $(S,A,I)$-sequences 
$$(\mathfrak{B}_0,\textbf{a}_0, \mathfrak{B}_1,
\textbf{a}_1,\dots ,
\mathfrak{B}_{k-1},\textbf{a}_{k-1},\mathfrak{B}_k)$$
such that $\mathfrak{B}_0 = G(    (\epsilon, F(\epsilon)) )$ and the
following conditions hold for each $i$ such that $0\leq i\leq k-1$:
%
%
%
%
%
%
%
%
%
%
%
%
%
%
%
%
%
\begin{enumerate}
\item
$\mathbf{a}_i = \Bigl(\ f_j\bigl(\, (\mathfrak{B}_0,\textbf{a}_0, \mathfrak{B}_1,
\textbf{a}_1,\dots ,\textbf{a}_{i-1},\mathfrak{B}_{i})\, \bigr)\ \Bigr)_{j\, \in\, I}$
\item
$\mathfrak{B}_{i+1}
= G\bigl(\, \bigl((\mathfrak{B}_0,\mathbf{a}_0,\dots , \mathfrak{B}_i,\mathbf{a}_i),
\ F((\mathfrak{B}_0,\mathbf{a}_0,\dots , \mathfrak{B}_i,\mathbf{a}_i))\, \bigr)\, \bigr)$.
\end{enumerate}

\medskip

\medskip

\noindent
Also the empty sequence is a finite evolution.
\emph{Infinite evolutions} are defined in the analogous way to be infinite sequences
$(\mathfrak{B}_0,\textbf{a}_0,\mathfrak{B}_1,
\textbf{a}_1,\dots )$ of the ordinal length $\omega$ and
satisfying the above conditions $1$ and $2$ with $\mathfrak{B}_0 = G((\epsilon, F(\epsilon)) )$.

\medskip

\medskip

\noindent
If $\mathcal{B} = (S,F,G,(f_i)_{i\in I})$ is a system and $E$ a structure-ended $(S,A,I)$-sequence,
then $(\mathcal{B},E)$ is called an \emph{instance}. If $E$ is also a 
finite evolution of the system, we may call $(\mathcal{B},E)$ a \emph{realizable} instance.
An instance (realizable or not) can also be called a \emph{pointed system} in analogy with
pointed models in modal logic. The last structure $\mathfrak{B}_k$ of $E$ is
called the \emph{current structure} or
\emph{current world} of $(\mathcal{B},E)$ (and also of $E$). The
set $S$ of $\mathcal{B} = (S,F,G,(f_i)_{i\in I})$ is
called the \emph{domain} or \emph{universe} of $\mathcal{B}$
(and also the domain of the system frame base $(S,F)$ and system frame $(S,F,G)$).

\medskip

\medskip

\noindent
Systems (and frames and 
frame bases) where all functions are total are called \emph{strongly regular}. We below
analyse systems, and occasionally ignore 
technically anomalous features arising in systems that are not strongly regular.

\subsection{On interpretations of systems}

\noindent
While there are numerous natural interpretations of systems as defined here, the following rather ambitious
interpretation stands out. A system frame base $(S,F)$ of a
system $(S,F,G,(f_i)_{i\in I})$ can be interpreted to represent the 
\emph{material} or \emph{physical} part of the system, while $G$ and the functions $f_i$ are 
the \emph{non-physical} or \emph{non-material} part.
The functions $f_i$ can indeed be considered to be individual \emph{agents},\footnote{The
functions $f_i$ encode behaviour strategies of agents and the indices in $I$ can be
thought to provide agent names or something of that sort, a unique name (or index) for each agent.
If desired, it is of course possible to
construct a physical counterpart (a body) for an agent and encode it into the structures in $S$.
The body need not necessarily be a connected or somehow local pattern. One natural choice is to
pick a new relation symbol $R_i$ for each agent index $i\in I$ to represent the body. But that is just one 
choice. The related function $f_i$ can in suitable cases be modeled by letting some part of the 
encoding (or body) of the agent encode, e.g., a Turing machine, possibly with some
fault tolerance included. The input to $f_i$ need not be encoded anywhere necessarily.
But if it is, then it is most naturally encoded by some
small, distinguished part of the current structure, suitably local to the body. This means that $f_i$
does not truly necessarily depend on the full sequence ending with the current
structure. Instead $f_i$ depends only on a crude representation of the actual input.
We will discuss these issues in a bit more detail below.} while $G$ can be
regarded as some kind of a high external controller---or perhaps \emph{chance} or \emph{luck}---that
determines the ultimate evolutive
behaviour of the system. The agents pick
actions from the set $A$, and based on the actions, $F$
determines a set of new possible worlds. 
The actual new world is then picked by $G$ from that set.
It is natural to consider $F$ to correspond to \emph{physical constraints}
within which the evolution happens, while $G$ is a more abstract (perhaps 
intuitively non-physical), chance-like entity.\footnote{It is worth noting that
interpretations of systems and the related 
metaphysical issues do not \emph{necessarily} have to be taken in some overtly literal sense. Interpretations 
can also be flexible frameworks that guide thinking in intuitive and fruitful ways. Moreover, it is 
worth remembering that systems also model various frameworks that can appear rather
concrete and even mundane, such as 
concrete games, simple physical systems, computations, et cetera. Nevertheless, the more
literal interpretation attempts are important as they relate to quite fundamental issues.}

\medskip

\medskip

\noindent
Within the collection of various interpretations, it is highly natural to consider 
systems where the tuples $\textbf{a}_i$ of agents' choices are
determined by the current structure $\mathfrak{B}_i$, as
opposed to entire 
sequence $(\mathfrak{B}_0,\textbf{a}_0,\dots ,
\mathfrak{B}_{i-1},\textbf{a}_{i-1},\mathfrak{B}_i)$
ending with $\mathfrak{B}_i$. 
This of course implies that for each $j\in I$, there exists a 
function $h_j$ such that 
$$f_j((\mathfrak{B}_0,\textbf{a}_0,\dots ,
\mathfrak{B}_{i-1},\textbf{a}_{i-1},\mathfrak{B}_i)) = h_j(\mathfrak{B}_i)$$
holds for every $i$.
Note that for each structure $\mathfrak{B}_i\in S$, the
function $f_j$ must be defined either on every
structure-ended sequence ending with $\mathfrak{B}_i$ or none of
such structure-ended sequences.\footnote{This is because the outputs of $f_j$ are
determined by the last structure of each input sequence. Thus also a possible lack of an
output is taken to be so determined.} Thus the domain of $h_j$ is precisely the
structures $\mathfrak{B}_i$ such that $f_j$ is defined on sequences ending with $\mathfrak{B}_i$.

\medskip

\medskip

\noindent
This reflects the idea that evolution histories---at least up to the extent that the
agents can see them---must be
encoded in the current structure, if anywhere. The current
structure could naturally represent, e.g., the physical world at
the current time instance, and the agents' behaviour would then be
assumed to depend only on the current physical world. Indeeed, even the full sequence 
$$(\mathfrak{B}_0,\textbf{a}_0,\dots ,
\mathfrak{B}_{i-1},\textbf{a}_{i-1})$$
can be partially (or even fully, within suitable situations) encoded into the
current world $\mathfrak{B}_i$ of the extended sequence
$$(\mathfrak{B}_0,\textbf{a}_0,\dots ,
\mathfrak{B}_{i-1},\textbf{a}_{i-1},\mathfrak{B}_i).$$
Obviously, different agents $f_j$ can be made to see (i.e., depend on) 
different (typically rather small) parts of that encoding.\footnote{Naturally
agents can also have a limited picture of the current world $\mathfrak{B}_i$. This 
issue will be discussed more later on below.}

\medskip

\medskip

\noindent
Also $F$ can be made dependent upon the last structure only. This is perhaps natural when $F$ is
interpreted to be the part of the physical nature that is not dependent upon chance. Then it may be
natural that all such past time events that are considered to affect $F$, should be readable (and
thus encoded into) the current structure.

\medskip

\medskip

\noindent
In contrast to $f_j$ and $F$, it is typically most natural (but of course optional) to 
keep the behaviour of $G$  
dependent on full input tuples (which are of type $((\mathfrak{B}_0,\mathbf{a}_0,
\dots , \mathfrak{B}_i,\textbf{a}_i),W)$ for $G$). This is natural if $G$ is
interpreted to be some kind of a pure luck factor or something
similar, a high external force or so on. Then it can be reasonable that
the output of $G$ is not readable from the concrete
current physical world but can be arbitrary, which in this case means simply dependence
upon the full history of structures and choices (and the set $W$).\footnote{
The article \cite{a18k} defines systems
according to the intuition that indeed only $G$ depends on full sequences.
We note here that there is an obvious typo in \cite{a18k}. There we should have
\begin{enumerate}
\item
$\mathbf{a}_i = (f_j(\mathfrak{B}_i))_{j\in I}$
\item
$\mathfrak{B}_{i+1}
= G\bigl(\, (\mathfrak{B}_0,\mathbf{a}_0,\dots , \mathfrak{B}_i,\mathbf{a}_i)\, \bigr)$,
\end{enumerate}
while the typo version has the first line $\mathbf{a}_i = (f_i(\mathfrak{B}_i))_{i\in I}$,
which is obviously wrong.
}

\medskip

\medskip

\subsection{Eliminating features\hspace{0.3mm}}

\medskip

\medskip

\noindent
It is worth noting that for conceptual
reasons, it is nice to have both $F$ and $G$ in systems, although the combined
action of $F$ and $G$ is essentially a single partial function. We could 
define systems differently, of course.
It is also worth noting that history features can often be relatively
naturally simulated in 
current structures by using suitable encodings.
This bears a resemblance to, e.g.,
defining tree unravelings in temporal logic, where each
node then fully determines the history of that node.

\medskip

\medskip

\noindent
Furthermore, we can make some of the functions $F$, $G$ and $f_j$
concrete (or perhaps physical) in the sense that
some or all of their features get encoded in the structures $\mathfrak{B}_i$.
Indeed, we already mentioned this possibility in relation to agent functions.
For example, we can indeed encode Turing machines into the
structures in system domains. The Turing machines are then required to fully indicate how
the concretized functions would operate.

\medskip

\medskip

\noindent
Let $(S,F,G,(f_i)_{i\in I})$ be a system and $h\in \{f_i\}_{i\in I}$ a 
concretized function. Suppose that each structure in $S$ encodes $h$
using some distinguished relation symbols $R_{h,j}$. For 
simplicity, suppose $h$ always depends only on
the current structure instead of the full history.
Now, the relations $R_{h,j}$ are required to ``output''
the same choices in each sequence ending with $\mathfrak{B}$ as 
what $h$ would output with the input $\mathfrak{B}$. Of 
course it is natural to make $h$ depend only on some small part of $\mathfrak{B}$, a part
that could be encoded close to where the relations $R_{h,j}$ have tuples. Closeness
here can be measured in relation to some binary distance relation $R$.
This makes
the facts $R_{h,j}(b_1,\dots , b_l)$ (here $b_1,\dots , b_l$ are elements of $\mathfrak{B}$) 
correspond to the material body of the agent $h$. Note that while we assumed $h$
depends only on current structures, we could encode
history features into structures for $h$ to see.

\medskip

\medskip

\noindent
Suppose we encode a concretized agent function $h$ into the model domains,
and suppose we also somehow encode the body of
the related agent. It is then natural (but of course
not necessary) to let
the body of the related agent contain the tuples encoding $h$. It is also natural (but not necessary) to
make the body local, as discussed above. When considering encodings, it is
worth noting that 
tuples of relations (in standard first-order models)
do not have a clear identity that carries from a model to another.
Indeed, if we have  a relation with two tuples, and the model changes so that in the
new model we again have two tuples but now somewhere else in the model, then
there is no obvious way of telling which new
tuple corresponds to which old tuple---if there is any intended correspondence in the first place.
If we wish to encode identities for tuples (in first-order models), one idea is to use ternary relations to encode
binary relations, with the first coordinate providing an indentity for the tuple. For example, a
fact $R(b_1,b_2,b_3)$ would correspond to a tuple encoding the 
pair $(b_2,b_3)$ and having $b_1$ as its indentity.

\medskip

\medskip

\noindent
As we have noted,
perceiving only a part of the current model is natural for agents, and it is
natural if the perceived part is in the vicinity of the material body of the agent. 
Next we discuss issues related to perception, and beyond.

\subsection{Partial and false information}\label{falsepartial}

Generally agents make their choices based on sequences 
$$(\mathfrak{B}_0,\textbf{a}_0,
\dots , \mathfrak{B}_{k-1},\textbf{a}_{k-1},\mathfrak{B}_k).$$
In other words, the functions $f_j$ are functions of such sequences.
The setting where all agents $f_j$ depend on the 
current structure $\mathfrak{B}_k$ only (i.e., the last structure of the
input sequence) will be below referred to as the \emph{positional scenario}. The general 
setting is refferred to as the \emph{general scenario}.

\medskip

\medskip

\noindent
In the general scenario, it is natural that agents $f_j$ do not use the full
sequence $$(\mathfrak{B}_0,\textbf{a}_0,
\dots , \mathfrak{B}_{k-1},\textbf{a}_{k-1},\mathfrak{B}_k)$$ leading to the 
current structure $\mathfrak{B}_k$, but instead some \emph{representation} of
that full sequence. 
Similarly, in the positional scenarion, it is  
natural to assume that the agents only see some \emph{representation} of $\mathfrak{B}_k$.

\medskip

\medskip

\noindent
In both scenarios, the representation may not necessarily resemble the 
represented sequence/structure at all, but could instead be partially or even wholly different.
The intuition of the representation is that it is the \emph{mental model} the
agent has about reality. Let us make this precise.

\medskip

\medskip

\noindent
We first consider the positional scenario. Fix a system $(S,F,G,(f_i)_{i\in I})$. 
While the functions $f_i$ can indeed depend on all of the current model $\mathfrak{B}_k$, which can be
quite reasonable when modeling perfect information games, it is highly natural to define
\emph{perception functions} to cover the scenario of partial and even false information.
Perception functions will make the agent functions $f_i$ depend upon \emph{perceived 
models} or \emph{mental models}. We let a  perception function for agent $i$ to be a
map $p_i:\ S \rightarrow S_i$, where $S_i$ is a class of
structures whose signature may be different from those in $S$. The class $S_i$ is
the class of mental models of agent $i$. We then
dictate that $f_i(\mathfrak{B}) = d_i(p_i(\mathfrak{B}))$
for each input $\mathfrak{B}\in S$, where $d_i:S_i\rightarrow A$ is
called the \emph{decision function} of agent $i$, and $A$ is simply the set of actions of the system we
are considering.

\medskip

\medskip

\noindent
For a concrete example, $p_i$ 
could be a first-order reduction, more or less in the sense of model 
theory or descriptive complexity, giving a very crude, finite approximation of the original model
(which is the input to $p_i$).\footnote{Here the output of $p_i$ approximates the 
current model. We note that it is often sensible to let each output of $p_i$ be intended to
approximate the current instance $(\mathcal{B},E)$, or even more, with nested beliefs, intentions of
agents, and so on.}
Now, even if the input model to $p_i$ is infinite, the output model can be finite and depend only on
some small part of the input model.\footnote{That part could indeed quite naturally be
mostly in the vicinity of the encoded body of the agent.}
Note that parts of the agents' epistemic states can be encoded into the original models
in $S$. Thus the agents can try to take into
account those parts of the other agents' epistemic states that they believe to have access to. How
much agent $i$ knows about the other agents' epistemic states in $\mathfrak{B}\in S$
will be reflected in the structure of the mental model $p_i(\mathfrak{B})$.\footnote{The 
mental model can reflect the agent's beliefs about the other agents' mental states, and the agent's beliefs about beliefs
about beliefs, and so on, possibly in a way that includes all agent-mixed nested modalities. But, of course, a
mental model does not have to try to do too much. Concerning modalites about nested beliefs, it is 
typically unrealistic to have everything in the mental model. Indeed, concerning information in 
general, it is very much realistic to have somehow
strongly partial (and perhaps false) information in the mental model. This relates directly to, e.g., 
limited memory capacities as well as limited perception.}
But of course this information can be highly partial, even false,
and obviously each agent tends to see different parts of 
information of the other agents' epistemic states. Of course agents do not
even have to know the full set of agents operating in the framework.

\medskip

\medskip

\noindent
The case in the general scenario is very similar and analogous to the positional scenario.
Many conceptual issues are more or less the same to a large extent. The 
difference in the formalism is that now $p_i$ maps from 
the set of structure-ended sequences of the original system into the set $S_i$ of
mental models of the agent $i$. The mental model can, in both the general and
positional setting, encode how much the agent $i$ remembers and understands about the
sequence that has lead to the current model in $S$. In the general scenario, however,
the mental model that $p_i$ outputs indeed formally depends on full input sequences, while
in the positional scenario, the sequence leading to the
current model is available only to the (possibly nonexisting) extent that 
the sequence is encoded in the current model.

\medskip

\medskip

\noindent
Different agents $i$ can of
course have different sets $S_i$.
But, in general, what should the mental models in the sets $S_i$ look like?
One option is that they encode sets of models in $S$. Such a set 
corresponds to the models in $S$ that the agent considers possible.\footnote{A related possibility would be to let a mental model encode a 
full set of instances $(\mathcal{B},E)$ the agent considers possible.}
This is a very classical approach.
It is completely unrealistic in many scenarios, as the agent would simply have too 
much information. Furthermore, it requires that all the models that the agent consider 
possible are actually models in $S$.

\medskip

\medskip

\noindent
A somewhat more realistic scenario goes as follows. A mental model in $S_i$ is
simply a set $\mathcal{A}$ of axioms in some logic. Intuitively, it axiomatizes
what the actual current model (and the history leading to it) 
should look like. It can also describe what the full global system (including possible futures,
the other agents and their mental models, the location  of the current
model, et cetera) looks like.\footnote{The picture of reality is indeed typically highly partial.}
We here concentrate mainly on how well the current model is known.
The set $\mathcal{A}$ could now contain the following.

\begin{enumerate}
\item
A set $\mathcal{F}$ of \emph{facts}.\footnote{We note that facts do 
not have to be true in any sense. Perhaps \emph{atoms} would be a
better term, or \emph{assumptions}.} These are atoms $R(b_1,\dots , b_k)$.
The elements $b_1,\dots , b_k$ are taken from some set $B'$ (which does not
have to be the domain of any model in $S$). Intuitively,
the agent could regard $b_1,\dots , b_k$ to be domain elements of the actual 
current model (which formally is the model $\mathfrak{B}$
such that $p_i(\mathfrak{B}) = \mathcal{A}$). The relation symbol $R$ can intuitively
belong to the signature of the models in $S_i$. 
Thus $R(b_1,\dots , b_k)$ could be
for example the fact $\tt{TallerThan}(\tt{John},\tt{Jack})$ representing the 
agent's belief that $\tt{John}$ is taller than $\tt{Jack}$ in the current actual
model\footnote{
$\tt{Jack}$ and $\tt{John}$ are
both elements of $B'$ (but need not really be anything in $\mathfrak{B}$, although it is
natural if they are). It is worth noting that generally the elements in $B'$ can be
differentiated---if desired---from the possible constant symbols in the signature of mental models.
For example, one may wish to keep the elements in $B'$ identical to supposed actual elements, 
while constant symbols are simply names of supposed actual elements.
We note that $\tt{Jack}$ and $\tt{John}$ 
here are not meant to be agents (although they could  possibly be). Instead, they are simply what the agent $i$
considers to be elements of $\mathfrak{B}$.} $\mathfrak{B}$.

\medskip

The relation $R$ can be something the agent \emph{considers} to somehow be an
actual relation in $\mathfrak{B}$, but $R$ can also be
some relation internal to the thinking of the agent. In that case also the 
elements $b_1,\dots , b_k$ can perhaps represent something that the agent
does not consider belonging to $\mathfrak{B}$. Indeed, such a virtual or
purely mental category of facts can be
very important. It could be desirable to include,
e.g., beliefs about 
other agents' beliefs into mental models.
This will involve encoding related issues into facts in $\mathcal{F}$.
\item
A set $\mathcal{F}'$ of negative facts. These are fully analogous to facts in $\mathcal{F}$, but
represent beliefs that the agent \emph{thinks} false. Formally,
these are literals $\neg R(b_1,\dots, b_k)$,
where $R(b_1, \dots, b_k)$ is as described above. Note that there is no problem if the 
agent holds a fact in $\mathcal{F}$ and its negation in $\mathcal{F}'$. Then the agent simply
has contradictory beliefs. It may be difficult for the agent to detect the contradiction.
\item
A set $\mathcal{B}$ of other axioms.\footnote{Not to be confused 
with the set $\mathcal{B}$ of an instance $(\mathcal{B},E)$.}
These are, in the most obvious cases, statements that the agent thinks 
the actual current model satisfies. They could also be statements about more abstract issues that are
not (necessarily) directly related to the current model,
for example statements about the beliefs of other agents. The \emph{only} difference between these and
the facts and negative facts in $\mathcal{F}\cup\mathcal{F}'$ is that these need not be literals.
These non-literals can still, of course, make use of the elements in $B'$, if desired. Again the agent
can have contradictory beliefs, as some subset of $\mathcal{B}$ can have a contradiction as a 
logical consequence. It could simply be difficult for the agent to deduce that contradiction. Or even, it is
possible that the agent later on does easily deduce that contradiction, but at this
stage of evolution, the agent has not yet been able to obtain the contradiction.
Such a situation occurs even in mathematical proofs; we typically do not immediately obtain a 
contradiction, but it takes some effort.
\end{enumerate}

\medskip

\medskip

\noindent
To give an example of the above scenario, let the system domain $S$ consist of
first-order models. Let the set $B'$ be the union of the domains of 
the models in $S$. Suppose the
current model $\mathfrak{B}\in S$ consists of a domain $\{a,b\}$ and a relation $R = \{(a,a),(a,b)\}$.
Let the mental model $p_i(\mathfrak{B})$ be
given by
$$\mathcal{F} = \{ R(a,a) \},\ \mathcal{F}' = \{\neg R(b,a)\},\ \text{ and }\ 
\mathcal{B} = \{\neg \exists^{\geq 8}x(x=x)\}.$$
We are here discussing a scenario where the mental model simply tries to
identify $\mathfrak{B}$ to the best possible extent.
The agent knows that $R(a,a)$ and $\neg R(b,a)$
as well as $\neg \exists^{\geq 8}x(x=x)$
hold, but the agent has no idea about whether---for example---the fact $R(b,b)$ holds or
whether there are more than two elements.
The agent knows, we suppose in this scenario, that the
actual model is one of the models in the set of $\{R\}$-models  
that satisfy $\mathcal{F}\cup\mathcal{F}'\cup\mathcal{B}$
and have domain $D$ such that $\{a,b\}\subseteq D\subseteq B'$.\footnote{If $\mathfrak{M}$ is
an $\{R\}$-model \emph{and} has $a$ and $b$ as domain elements, then  we define
that $\mathfrak{M}\models\mathcal{F}\cup\mathcal{F}'\cup\mathcal{B}$ if the
expansion $\mathfrak{N}$ of $\mathfrak{M}$ with constant symbols $a$ and $b$
(interpreted such that $a^\mathfrak{N} = a$ and $b^\mathfrak{N} = b$)
satisfies all formulae in $\mathcal{F}\cup\mathcal{F}'\cup\mathcal{B}$. (Note that $\neg R(a,b)$ already implies
that there must be two elements at least ($a$ and $b$ are different elements), and note indeed that
we do not even interpret these formulae on models without $a$ and $b$ in the domain. Of course one
could avoid all this, if desired, and work only with the usual conventions concerning constant symbols.)}
Thus the setting resembles open world querying.
Now, to fully know the model $\mathfrak{B}$, the mental model could be given by
\begin{multline*}
\mathcal{F} = \{ R(a,a), R(a,b) \},\ \mathcal{F}' = \{\ 
\neg R(b,a), \neg R(b,b)\ \}\\ \text{ and }\mathcal{B} = \{\exists^{=2} x   (x=x)\}.
\end{multline*}
Note that here we give the full relational diagram of $\mathfrak{B}$ \emph{and} specify that
there are no more elements than those mentioned in the diagram. This suffices to 
fully specify the model in this case.\footnote{Here we did not include atoms $a=a$ in the diagram, 
but of course one would generally have to include them to always be able to tell what the
domain is when looking at the full diagram.
When diagrams indeed mean sets of literals where the constants in the literals 
are domain elements, we can specify models fully with 
suitable diagram notions, not only up to isomorphism, if we so wish for one reason or another.
But there is nothing technically deep behind this, 
and different conventions are possible for
different ways of modeling.}

\medskip

\medskip

\noindent
Now, the agent $i$ must choose an action based on the mental model $p_i(\mathfrak{B})$.
This is done via a function $d_i:S_i\rightarrow A$ that maps mental models to actions in $A$.
Now, a typical agent has limited reasoning resources, not being logically omniscient.
Indeed, as we have discussed, it could even in some cases be difficult for the agent to deduce a 
contradiction from a fact in $\mathcal{F}$ and its negation in $\mathcal{F}'$. This is
even typical if $\mathcal{F}$ and $\mathcal{F}'$ are large (physical) look-up tables.
And deducing a contradiction from a contradictory set $\mathcal{B}$ is 
likewise not always straightforward.

\medskip

\medskip

\noindent
One natural way to model $d_i$ is to use the limited reasoning capacities 
described in \cite{a17k2}. The idea is
that the agent uses logical reasoning,
but has access only to a possibly too small collection of inference rules and may also have to truncate
reasoning patterns after quite short reasoning chains. The premises consist of
the set $\mathcal{F}\cup\mathcal{F}'\cup\mathcal{B}$. It is natural for example to
impose a fixed limit $n$ dictating how many times the agent is allowed to use the inference
rules. Also, it is natural to put similar limitations onto the set of formulae the
agent can know at any time. So, if the agent reasons starting from 
$\mathcal{F}\cup\mathcal{F}'\cup\mathcal{B}$, the agent cannot add new formulae into
the setting without a limit when reasoning. The
agent may have to throw some formulae away during the reasoning process. While this 
mainly models finite memory capacities, note, however, that of course the agent could use
external look-up tables to store information. But those could, on the other hand, become
large and slow to read.
Anyway, in an ideal case, the agent can
deduce the full structure of the current model $\mathfrak{B}$ 
based on the mental model, and perhaps even the 
full history leading to $\mathfrak{B}$, and beyond, all the way to the global 
features of the system. If the agent $i$ can always deduce the full history,
then $f_i$ can depend on full histories.

\medskip

\medskip

\noindent
It is obviously dependent upon the agent (and even the current instance)  what 
reasoning tools can be used, and how complex reasoning patterns are allowed. Concerning
reasoning tools, it is reasonable to add inference rules to the set 
$\mathcal{F}\cup\mathcal{F}'\cup\mathcal{B}$. An additional set $\mathcal{I}$ could be used.
There should be ways to modify the set $\mathcal{I}$ based on the current world and the history.
Such ways can be encoded into the function $p_i$ that produces the mental models. A
later mental model is typically dependent upon an earlier one, e.g., $p_i(\mathfrak{B}_{j+1})$ upon
and $p_i(\mathfrak{B}_{j})$; this dependence could be mediated via the actual 
world $\mathfrak{B}_{j+1}$.

\medskip

\medskip

\noindent
Of course one does not have to use standard logic to model truncated and limited
reasoning, but also, e.g., complexity classes and computation
devices with suitably limited capacities. The mental models above are a
starting point, but of course one would like to add more general features to the picture.
For example, probabilistic
and fuzzy features (e.g., probabilistic weights on the literals and even general axioms)
are surely interesting. And obviously probability theory is not likely to suffice, but generalizations are needed.
Other approaches that also immediately suggest themselves
include using neural networks and other frameworks that involve possibilities for heuristic reasoning.
The obvious places where to use neural networks
concern the perception and decision functions $p_i$ and $d_i$.
A neural network device would be a natural option for producing the outputs of $p_i$. It would
look at some small part of the current model (and perhaps its history) and operate based on that.
Also $d_i$ could quite naturally be computed,
based the mental model, via a neural network device. We could even remove the mental
model from between $p_i$ and $d_i$ altogether, if desired. However, concerning 
human agents, it would
ultimately perhaps be more informative to combine
the use of neural networks with more classical features.

\medskip

\medskip

\noindent
It is worth noting that in our concrete 
example of a mental model, the set $\mathcal{F}\cup\mathcal{F}'$
approximated a first-order model.
But human agents more typically entertain picture-like
representations of models, that is, drawings of structures
rather than the structures themselves. It can be difficult to detect, e.g., graph isomorphism.
To account for more geometric mental models, we could
modify the $\mathcal{F}\cup\mathcal{F}'\cup\mathcal{B}$ approach a bit. 
The idea is to add three-dimensional
grids to the setting. Let $G_1,\dots , G_{\ell}$ be such grids.\footnote{These are models with 
three binary relations, $H$ indicating the left-to-right neighbour relation, $V$ indicating
the down-to-up neighbour relation, and $D$ indicating the closer-to-the-viewer relation. 
The relations are analogous to 3D coordinate axis orientations.}
We let each grid have a finite (but perhaps large)
domain and thus correspond to a finite set of points forming a rectangular cuboid array.
Now, we identify each (or alternatively, some) of the elements $b\in B'$ appearing
in the literals of $\mathcal{F}\cup\mathcal{F}'$ with some grid point. If there are
elements in the formulae of $\mathcal{B}$ that do not occur in the literals, then those
elements can also be identified with grid points. It is natural (but not necessary) to 
require that each literal has its elements in a single grid. Now the
patterns described via $\mathcal{F}\cup\mathcal{F}'$ have become
geometric objects (we draw the tuples into the grids in the obvious way). We have
drawings in three dimensions (and these could be made two dimensional as well).
The reason we have
started with several rather than a single grid is that typically an agent
entertains a collection of mental images rather than a single one.

\medskip

\medskip

\noindent
It is interesting to note that while $\mathcal{F}\cup\mathcal{F}'$ corresponds to 
knowledge, $\mathcal{B}$ in some sense relates to understanding, or at least more
abstract knowledge. We could add a set $\mathcal{C}$ to $\mathcal{F}\cup\mathcal{F}'\cup\mathcal{B}$, 
this being a set of suitably encoded reasoning algorithms that the
agent could then use on the formulae in $\mathcal{F}\cup\mathcal{F}'\cup\mathcal{B}$ 
and their more or less immediate logical consequences.
$\mathcal{C}$ could contain at least some proof rules (as the set $\mathcal{I}$
discussed above did). Now $\mathcal{C}$
would relate quite nicely to
understanding and the look-up-table-like set $\mathcal{F}\cup\mathcal{F}'$ to
knowledge. Of course somehow truncated reasoning, not 
full logical consequence, would be natural. Indeed, full logical consequence seems to relate to 
potential \emph{knowability} rather than knowledge.

\medskip

\medskip

\noindent
Summarizing this section so far, we have identified ways to model partial information and even
\emph{false information} via mental models given by $p_i$.
A partially false and strongly incomplete picture is a reasonably
natural starting point for modeling attempts.
We have also discussed how $d_i$ could take into account limitations in
reasoning capacities. There are many ways to do this, and obviously a huge
range of issues to investigate.

\medskip

\medskip

\noindent
So far we have concentrated on the positional scenario.
In the general scenario, however, the functions $p_i$ and $d_i$
are \emph{very much conceptually analogous} to their counterparts in the positional scenario, so 
the above investigations also largely apply conceptually in the general setting.
Formally, the domain of $p_i$ is the set of structure-ended $(S,A,I)$-sequences and
the output is a mental model. It is perhaps most natural to make the domain of $d_i$
simply the set of mental models, as in the positional scenario. Indeed, even in the positional 
scenario, some parts of histories would often become encoded in the mental models, as there can be
history features encoded in a current structure.
However, it can also be reasonable to, e.g., let $d_i$ depend on the chain of mental models leading to the 
current one, and the agent's own action in every past iteration step.

\subsection{Further issues}

\noindent
Systems can be used to model games, computation and physics systems, to 
name a few possibilities. Indeed, all kinds of interactive scenarios are reasonably naturally 
modeled by systems. Concerning applications in physics systems, the advantage of our formal
systems is the possibility of concretely modeling supposed mental entities (agents and $G$) together
with the supposably physical part (structures and $F$).\footnote{We note that the division ``agents and $G$'' vs
``structures and $F$'' does not necessarily provide a strict gap between what would be 
conceived as mental and what physical. Indeed, of course the supposed mental and physical 
realms are likely to show some connection between them to enable interaction between the realms.
The agents realistically have perception functions $p_i$ via 
which they see the structures in the system domain.
And the function $F$ looks at the \emph{actions} of the agents and 
provides an output partially based on that. However, we coud assert, e.g., that $F$
only sees the actions once they have been performed, making $F$ fully material in some sense.}

\medskip

\medskip

\noindent
Cellular automata provide a starting point for digital physics, but systems, as defined above, are much 
more flexible.\footnote{Of course one of the most obvious ideas is to make 
functions computable or semi-computable. But it is interesting to keep also
more general functions in the picture, for example it could be
quite natural to let $G$ be uncomputable. And it is often natural to let $F$
output infinite sets.}
The metaphysical setting of systems
provides a lot of explanatory power for understanding phenomena.\footnote{Indeed, it is natural to
regard systems as a framework providing a formal metaphysical setting for
modeling seemingly less fundamental frameworks with more contingent properties,
such as particular physical processes, for instance.}
The way the supposedly mental constructs ($G$ and each $f_i$) interact with the material 
parts is highly interesting. As systems are fully formal, concrete modeling attempts
will force new concepts and insights to emerge. 

\medskip

\medskip

\noindent
One of the most concrete and obvious advantages of systems when compared to, e.g.,
standard cellular automata, is
that it is not necessary to keep agents (and other entities) local. Furthermore, it is
not necessary (although can be natural) to keep agents and other entities computable.
However, computability and semi-computatbility are obviously very important 
issues. As suggested in \cite{a15k2}, extensions of the Turing-complete logic $\mathcal{L}$ can be
naturally used as logics to guide systems. We will discuss this issue below in Section \ref{systems and logic}.

\subsection{More general systems}

Our notion of a system can of course be generalized. Indeed, currently
every current structure has a finite history leading to it. To allow for infinite
past evolutions, and to get rid of the discreteness of the steps between subsequent 
models, we define the following notion.

\medskip

\medskip

\noindent
A \emph{total g-system} (g for general\footnote{A \emph{g-system} is
defined to be a system that can be
obtained from a total g-system by allowing some of the involved
functions to be partial.}) is 
defined to be a tuple $$(S,(R_j)_{j\in J},F,(f_i)_{i\in I},G)$$ such that the following conditions hold.

\begin{enumerate}
\item
$S$ is a set of structures.
\item
Each $R_j\subseteq S^{k_j}$ is a $k_j$-ary relation
over $S$. Intuitively, $R_j$ could for example give a partial order of the structures in $S$ that
corresponds to time. But of course other interpretations are possible.
%
%
\item
$F$ is a function $\mathcal{P}(S)\times A^I\ \rightarrow\ \mathcal{P}(\mathcal{P}(S))$.
Intuitively, $F$ maps each history (a set of
structures in $S$) to a set of extended 
evolutions (a collection of subsets of $S$). The output depends also on the 
actions of the agents.
\item
Each $f_i$ is a function $\mathcal{P}(S)\rightarrow A$ from histories to actions.
\item
$G$ is a function $\mathcal{P}(S)\times A^I\ \rightarrow\ \mathcal{P}(S)$
such that $G(t)\in F(t)$ for each input where $F(t) \not = \emptyset$.
If $F(t) = \emptyset$, then $G(t) = \emptyset$. Intuitively, $G$ just picks the 
actual outcome from the set of possible outcomes given by $F$.\footnote{We could of
course combine the actions of $F$ and $G$ and thereby only have one function, but it is
nice and natural to include both of them.}
\end{enumerate}

\noindent
This is a relatively general approach. For example, it is possible to cover 
cyclic approaches to time, even dense ones, for example by 
mapping $S$ into $\mathbb{R}^2$. And of course one can 
consider approaches with no time concept in the first place.

\medskip

\medskip

\noindent
A basic notion in total $g$-systems is a set of structures. The principal intuition of
such a set is a history of some kind. Note that histories do not this time contain actions, so single action tuples in $A^I$ are
perhaps most naturally continuous processes acting all the way through the input (a history). But of course actions
could be embedded into histories in a different way, leading to generalizations.
Another one of the reasonable further generalizations is to
base actions on sets of histories instead of a single one. This leads to 
systems $(\mathcal{P}(S),(R_j)_{j\in J},F,(f_i)_{i\in I},G)$ with the following specification.\footnote{
The specification is close to simply replacing $S$ in the previous specification by $\mathcal{P}(S)$,
but not exactly the same.}

\begin{enumerate}
\item
$S$ is a set of structures.
\item
Each $R_j\subseteq S^{k_j}$ is still simply a $k_j$-ary relation
over $S$.
%
%
\item
$F$ is a function $\mathcal{P}(\mathcal{P}(S))\times A^I\
\rightarrow\ \mathcal{P}(\mathcal{P}(\mathcal{P}(S)))$.
\item
Each $f_i$ is a function $\mathcal{P}(\mathcal{P}(S))\rightarrow A$.
\item
$G$ is a function $\mathcal{P}(\mathcal{P}(S))\times A^I\ \rightarrow\ \mathcal{P}(\mathcal{P}(S))$
such that $G(t)\in F(t)$ for each input where $F(t) \not = \emptyset$.
If $F(t) = \emptyset$, then $G(t) = \emptyset$. 
\end{enumerate}

\noindent
Further generalizations would involve, e.g., putting weights on structure sets and
sets of structure sets. And so on and so on.

\medskip

\medskip

\noindent
A highly general setting to model nested beliefs can be based on 
the concurrent game models of Alternating-time temporal logic. Consider the 
reasonably flexible concurrent game models as defined in, inter alia, \cite{aamas}.
These can be given canonical tree unravelings; we begin from a 
single state and unravel from there. This gives an unraveled model with a root.
We let $\mathcal{T}$ be a set of
such unravelings. One of the unravelings could be what is actually happening, but we will
also model (possibly false and even unrealizable) beliefs about time flow.

\medskip

\medskip

\noindent
Now, each state of $\mathcal{T}$ has a
unique history.
 (Recall that a state is now a copy of a state in
some original model, but also with a unique history.) Given the
set of agents is $K$, suppose there is, for each $k\in K$, a 
binary relation $R_k\subseteq Q\times Q$, where $Q$ is the set of 
states of $\mathcal{T}$. Intuitively, $(q_1,q_2)\in R_k$ if in
the state $q_1$, the agent $k$ considers it possible that (s)he is
currently in state $q_2$. So these are epistemic relations, and 
they can point from one unraveling to another. The nice 
thing here is that the states have a unique history, so the 
binary relations are also binary epistemic
relations over the set of histories. And each history has a sequence of 
changing beliefs about the current history, et cetera.

\medskip

\medskip

\noindent
Now we can analyse interesting nested beliefs that also involve temporal
statemens. Suppose $k$ is at $q_a$ and $R_k$ points only to $q_b$
from $q_a$. Now $k$ believes to be at $q_b$. Suppose the predecessor of $q_b$ is $q_c$.
Now $k$ thinks the previous state was $q_c$. Now suppose $q_c$ is also
the predecessor of $q_a$. Then $k$ is right about the
previous state but for a wrong reason.\footnote{All kinds of questions rise about the
setting, the epistemic relations, and so on. For example, one would typically---but perhaps not
always---require the epistemic relations to be transitive. But, we
shall not discuss this setting in depth here.}

\medskip

\medskip

\noindent
Here we did not nest the beliefs of different agents, $k$ and $l$ for example. 
But that is of course possible, leading to beliefs about beliefs with a temporal dimension, and
so on and so on.
All this is nice and quite general. The mental model of $k$ at $q$ could be considered to be the set of
pairs $(\mathcal{T},q')$ such that $(q,q')\in R_k$.\footnote{If we want to remember the 
original models that gave rise to the unravelings, we get an interesting \emph{static-sameness
relation} over the set of states of $\mathcal{T}$ defined such that $q$ and $q'$ are related if
the last state in  $q$ and $q'$ is the same state (and originates from the same model).
Note that $q$ and $q'$ are indeed sequences formally.}
However, in the current article we are very much interested in
using (what would be) the \emph{internal} structure of states. In $\mathcal{T}$, the 
states do not have an internal structure.\footnote{They have no internal
structure unless we look into the sequences (in the 
original models) that define the states in $\mathcal{T}$ when we unravel models. But we do
not mean to look into them here (with the exception of defining the static-sameness relations).}
The setting of $\mathcal{T}$ uses epistemic relations that in a sense seem  blind to
the possible internal structures of states. Nevertheless, both the
internal view and the external one 
can be useful, and surely the approaches can be combined. Indeed,
states of $\mathcal{T}$ might as well be (replaced by)
relational structures, and conversely, our structure-based setting with systems and mental
models does suggest global epistemic relations for agents.

\medskip

\medskip

\noindent
In $\mathcal{T}$ and also generally in Kripke models, it is interesting to
define a metaphysical (rather than epistemic) relation $S_k$ for agents $k$. As $R_k$, this is
also a binary relation on states of $\mathcal{T}$ (or some other Kripke model).
Now, using the relations $S_k$ and $R_k$, it is at least relatively natural to define the 
standard indicative implication ``$\varphi$ implies $\psi$'' as $[ R_k ](\varphi\rightarrow \psi)$ and the subjunctive 
implication ``if $\varphi$ was the case, then so would $\psi$'' as $[ S_k ](\varphi\rightarrow \psi)$.
Here $[ R_k ]$ and $[ S_k ]$ are boxes with the accessibility relations $R_k$ and $S_k$, and $\varphi$
and $\psi$ are formulae whose truth sets are sets of states. The metaphysical modalities $S_k$ are
likely to be similar or even the same for all agents. Typically they would be equivalence relations, but not
necessarily always. Finally, it is worth noting that two states can naturally be coupled
with an \emph{indistinguishability relation} of agent $k$ if $k$ sees precisely the
same states from the two states via the epistemic relation $R_k$.



\section{Systems and logic}\label{systems and logic}

\noindent
The article \cite{tc} defines a natural Turing-complete extension $\mathcal{L}$ of
first-order logic $\mathrm{FO}$. This new logic is Turing-complete in
the sense that it can define precisely all recursively enumerable classes of finite structures.
The logic is based on adding two new capacities to $\mathrm{FO}$.
The first one of these is the capacity to modify models. The logic can \emph{add new points} to models
and \emph{new tuples} to relations, and dually, the logic can
\emph{delete} domain points and tuples from relations.\footnote{Strictly speaking, the 
system defined in \cite{tc} did not include the capacity to \emph{delete points} from model domains.
However, this possibility was briefly discussed, and it was then ruled out only due to page limitations in the paper. The reason for leaving out the capacity to delete domain points was mainly related to the fact that this can lead to 
variables $x$ whose referent has gone missing from the model domain. Also empty models appear.
However, in the current article we let $\mathcal{L}$ refer to the logic that also has the domain
element deletion operator (and the empty model is fine). Furthermore, \cite{tc} made the
some other limitations to the syntax of $\mathcal{L}$ so that a semantic game does not
lead to anomalous situations where again $x$ has no value (even if there are no
domain element deletions). Such situations were described to result in 
from \emph{non-standard jumps}. Here we impose no limitations on the syntax. Basically the 
result of these relaxations is simply more situations where
neither player has a winning strategy in the game. Also, domain element deletion is crucial in 
various scenarios that
allow all computable model transformations to be modeled directly.} 
The second new capacity is the possibility of formulae to refer to themselves.
The self-referentiality operator of $\mathcal{L}$ is based on a construct that enables \emph{looping}
when formulae are evaluated using game-theoretic semantics.\footnote{See \cite{tc} for sufficient details on 
game-theoretic semantics, and see
\cite{hinti}, \cite{lore} for some early ideas leading to the notion of game-theoretic semantics.}

\medskip

\medskip

\noindent
The reason the logic $\mathcal{L}$ is
particularly interesting lies in its \emph{simplicity} and its
\emph{exact behavioural
correspondence} with Turing machines.
Furthermore, it provides a natural and particularly simple
\emph{unified perspective} on logic and
computation. Also, the new operators of $\mathcal{L}$ directly
capture two fundamental classes of
constructors---missing from $\mathrm{FO}$---that are used all the time in everyday mathematics:

\begin{enumerate}
\item
fresh points are added to constructions
and fresh lines are drawn, et cetera, in various contexts in, e.g., geometry, and
\item
recursive operators are omnipresent in mathematical practice,
often indicated using the three dots (...).
\end{enumerate}

\medskip

\medskip

\noindent
One of the advantageous properties of $\mathcal{L}$ (in relation to typical logics) is that it can
indeed modify models. And models surely do not have be static, althought that is still the typical approach. Even in classical 
mathematics, we modify our structures. For example in compass-and-straightedge
constructions, we draw new points and lines. While there exist logics that modify 
structures (e.g., sabotage modal logic, some public announcement logics, et cetera), $\mathcal{L}$
offers a fundamental framework for modifications.

\subsection{The syntax and semantics of $\mathcal{L}$}

Here we  give the syntax and semantics of $\mathcal{L}$. For the 
full formal details, see \cite{tc}.
We let $\mathcal{L}$ denote the language that extends the
syntax specification of first-order logic by the following formula construction rules:
%

\vspace{-2mm}

\begin{enumerate}
\item
$\varphi\ \mapsto\ \mathrm{I}x\, \varphi$
\item
$\varphi\ \mapsto\ \mathrm{I}_{R(x_1, \dots , x_n)}\ \varphi$
\item
$\varphi\ \mapsto\ \mathrm{D}x\, \varphi$
\item
$\varphi\ \mapsto\ \mathrm{D}_{R(x_1, \dots , x_n)}\ \varphi$
\item
$C_i$ is an atomic formula (for each $i\in \mathbb{N}$)
\item
$\varphi\ \mapsto\ \, C_i\, \varphi$
\item
We also allow allow atoms $X(x_1,\dots , x_k)$ where $X\in\mathit{tsymb}$ is a $k$-ary relation
symbol not in the signature considered. The set $\mathit{tsymb}$ contains a countably infinite set of
symbols for each positive integer arity.\footnote{The name $\mathit{tsymb}$ comes from the fact that 
these symbols are analogous to Turing machine tape symbols, i.e., symbols not part of the input language.
It is conjectured in \cite{tc} that
the symbols in $\mathit{tsymb}$ are not needed for Turing-completeness of $\mathcal{L}$,
unless the background signature contains no symbols of arity at least two. The $R$ in the 
operators $\mathrm{I}_{R(x_1,\dots x_n)}$
and $\mathrm{D}_{R(x_1,\dots x_n)}$ can be a relation symbol in the signature or a
symbol in $\mathit{tsymb}$.}
\end{enumerate}
\noindent
Intuitively, a formula of type $\mathrm{I}x\, \varphi(x)$ states that it is 
possible to \emph{insert} a fresh, isolated element $u$ to the domain of the current model so
that the resulting new model satisfies $\varphi(u)$. The fresh element $u$ being \emph{isolated} 
means that $u$ is disconnected from the original model; the relations of the original model are
not altered in any way by the operator $\mathrm{I}x$, so $u$ does not become part of
any relational tuple at the moment of insertion. (Note that we assume a purely relational signature for
the sake of simplicity.)

\medskip

\medskip

\noindent
A formula of type $\mathrm{I}_{R(x_1, \dots , x_n)}\ \varphi(x_1,\dots , x_n)$
states that it is possible to insert a tuple $(u_1,\dots , u_n)$ to the relation $R$ so
that $\varphi(u_1,\dots , u_n)$ holds in the obtained model. The tuple $(u_1,\dots , u_n)$ is a
sequence of elements in the original model, so this time the domain of the model is 
not altered. Instead, the $n$-ary relation $R$ obtains a new tuple. The deletion 
operators $\mathrm{D}x$ and $\mathrm{D}_{R(x_1, \dots , x_n)}$ 
have obvious dual intuitions to the insertion operators.

\medskip

\medskip

\noindent
The new atomic formulae $C_i$ can be regarded as \emph{variables} ranging over
formulae, so a formula $C_i$ can be considered to be a \emph{pointer} to (or the
\emph{name} of) some other formula.
The formulae $C_i\, \varphi$ could intuitively be given the following reading:
\emph{the claim $C_i$, which states that $\varphi$, holds.}
Thus the formula $C_i\, \varphi$ is both \emph{naming} $\varphi$ to be called $C_i$
and \emph{claming} that $\varphi$ holds.\footnote{It is worth noting that the
approach in $\mathcal{L}$ to formulae $C_i\, \varphi$ bears \emph{some}
degree of purely \emph{technical} similarity to evaluations of fixed-point
operators of the $\mu$-calculus via game-theoretic semantics. However, that 
approach to fixed-point operators has not---to the author's knowledge---been connected to
self-referentiality and the related concepts in any way. Indeed, the approach of $\mathcal{L}$ is---to the
author's knowledge---conceptually novel, and has game-theoretic semantics as an
underlying primitive starting point. Furthermore, the
approach in $\mathcal{L}$ is
fully general and not \emph{explicitly} related to any fixed-point concepts. For example,
there are no monotonicity restrictions imposed on formulae, unlike in the $\mu$-calculus for example.
Another thing worth noting here is that \cite{tc} simply uses \emph{numbers} as
formula variables (which here are symbols $C_i$).}
Importantly, the formula $\varphi$ can contain $C_i$ as an atomic formula.
This leads to self-reference. For example, the liar's paradox
now corresponds to the formula $C_i\, \neg\, C_i$. .

\medskip

\medskip

\noindent
The logic $\mathcal{L}$ is based on game-theoretic semantics GTS which 
directly extends the standard GTS of $\mathrm{FO}$. Recall that the GTS of $\mathrm{FO}$ is
based on games played by the \emph{verifier} and \emph{falsifier}, or more accurately,
between \emph{Eloise} and \emph{Abelard}, Eloise first holding the verifying role (which can
change if a negation is encountered). In a 
game for checking if $\mathfrak{M}\models \varphi$, Eloise is trying to show (or 
verify) that indeed $\mathfrak{M}\models\varphi$ and Abelard is opposing this, i.e.,
Abelard wishes to falsify the claim $\mathfrak{M}\models\varphi$. The 
players start from the original formula and work their way towards subformulae and ultimately atoms. 
See \cite{tc} for further details concerning the GTS of $\mathrm{FO}$ and also $\mathcal{L}$.

\medskip

\medskip

\noindent
We now discuss how the
rules for the $\mathrm{FO}$-game are extended to deal with $\mathcal{L}$.
Further details are indeed given in \cite{tc}. Each game position involves a model -assignment
pair $(\mathfrak{M},f)$ and a formula $\psi$. The point of the
assignment $f$ is to give interpretations to the free
variable symbols of $\psi$ in the domain of $\mathfrak{M}$. 
A game position also specifies which one of Eloise and Abelard is the verifying player.
Furthermore, there is an assignment that gives interpretations of the relations $X$ not in the signature. In
the beginning of the game play, the relations $X$ are all empty relations, so they must be
built by adding tuples during the game play. For simplicity, we do not
explicitly write down this assignment for relations $X$ below, but instead assume it is somehow encoded into the 
models involved.\footnote{For example, we could assume that each $X$ in the formula we are evaluating is 
interpreted in the model we are investigating, being originally interpreted as the empty relation. But despite
that, the relations $X$ are not considered part of the official signature of the model.} The game rules go as follows.

\begin{enumerate}
\item
In a position involving $(\mathfrak{M},f)$ 
and the formula $\mathrm{I}x\, \psi(x)$, the game is
continued from a position with $(\mathfrak{M}',f[x\mapsto u])$ and $\psi(x)$, where $\mathfrak{M}'$ is
the model obtained by simply inserting a fresh isolated point $u$ to the 
domain of $\mathfrak{M}$. The fresh point is named $x$.
\item
In a position with $(\mathfrak{M},f)$ and $\mathrm{I}_{R(x_1,\dots , x_n)}\psi(x_1,\dots , x_n)$, 
the verifier chooses a tuple $(u_1,\dots , u_n)$ of
elements in $\mathfrak{M}$ and
the game is continued from the
position with $(\mathfrak{M}',f[x_1\mapsto u_1, \dots , x_n\mapsto u_n])$
and $\psi(x_1,\dots , x_n)$ where $\mathfrak{M}'$ is
obtained from $\mathfrak{M}$ by inserting the tuple $(u_1,\dots , u_n)$ to
the relation $R$. Note that $R$ can be part of the signature or one of the
relations $X$ outside the signature.
\item
Consider a position involving $(\mathfrak{M},f)$ 
and the formula $\mathrm{D}x\, \psi$. Now the game is
continued from a position with $(\mathfrak{M}',f\setminus\{(z,u)\, |\, z\in\mathrm{VAR}\})$
and $\psi$, where $\mathfrak{M}'$ is
the model obtained by deleting the point $u$ such that $f(x) = u$ from $\mathfrak{M}$ 
(and $\mathrm{VAR}$ is the set of all first-order variable symbols). If no such point $u$ exists, i.e., if $f$
does not have $x$ in the function domain, then nothing is done. Note that the
assignment function $f\setminus\{(z,u)\, |\, z\in\mathrm{VAR}\}$ is of course
obtained from $f$ by removing the pairs of type $(z,u)$ where $z$ is a variable. Thus, in 
particular, the pair $(x,u)$ is removed.
\item
In a position with $(\mathfrak{M},f)$ and $\mathrm{D}_{R(x_1,\dots , x_n)}\psi(x_1,\dots , x_n)$, 
the verifier chooses a tuple $(u_1,\dots , u_n)$ of
elements in $\mathfrak{M}$ and
the game is continued from the
position with $(\mathfrak{M}',f[x_1\mapsto u_1, \dots , x_n\mapsto u_n])$
and $\psi(x_1,\dots , x_n)$ where $\mathfrak{M}'$ is
obtained from $\mathfrak{M}$ by deleting the tuple $(u_1,\dots , u_n)$ from
the relation $R$. If there is no such tuple in $R$, then the relation stays as it is. As above, we
note that $R$ can be in the signature or one of the
relations $X$ outside the signature.
\item
In a position involving $(\mathfrak{M},f)$ and $C_i\, \psi$, we
simply move to the position involving $(\mathfrak{M},f)$ and $\psi$.
\item
In an atomic position involving $(\mathfrak{M},f)$ and $C_i$, the game
moves to the position $(\mathfrak{M},C_i\, \psi)$. Here $C_i\, \psi$ is a subformula of
the original formula that the semantic game began with. If there are many such 
subformulae $C_i\, \psi$, the verifying player can freely jump to any of them. If
there are no such formulae, the game play ends with neither player winning.\footnote{An 
alternative convention would be to jump to the immediately superordinate 
formula $C_i\, \psi$ in the cases where there are many choices. If no such 
immediately superordinate choice was available, the game play would end with 
neither player winning.}
%
%
%
\item
In a position with $(\mathfrak{M},f)$ and an atom of type $R(x_1,\dots , x_n)$ or $x= y$, the 
game play ends. We denote the atom by $\psi$ and note that $R$ can once again be in
the signature or one of the symbols $X$. The verifier
wins if $(\mathfrak{M},f)\models \psi$, where $\models$ is
the semantic turnstile of standard $\mathrm{FO}$. The falsifier
wins if $(\mathfrak{M},f)\models \neg \psi$. If $\psi$ contains any variables that are not in
the domain of $f$, then neither player wins.
\item
The positions involving $\exists$, $\wedge$, $\neg$ are dealt with exactly as in 
standard first-order logic.
\end{enumerate}

\medskip

\medskip

\noindent
Just like the $\mathrm{FO}$-game, the extended game ends only if an atomic
position with an atom $R(x_1,\dots , x_n)$ or $x = y$ is encountered.\footnote{Well, now
the game can end also if $C_i$ refers to no formula $C_i\psi$, but this is anomalous.}
Here $R$ can be in 
the signature or one of the relations $X$.
The winner is then decided precisely as in the $\mathrm{FO}$-game. That is, the
verifying player wins if the pair $(\mathfrak{M},f)$ in that position satisfies
the formula involved, and the falsifying player wins if $(\mathfrak{M},f)$ satisfies the
negation of the formula. In the anomalous unintended cases where $f$ does not interpret all of
the variables in the formula $R(x_1,\dots , x_n)$ or $x = y$ of the position,
neither player wins the play of the game.

\medskip

\medskip

\noindent
Since the play of the game can end
only if an atom $R(x_1,\dots , x_n)$ or $x = y$ is encountered, the
game play can go on forever, as for example the games for $C_i\, C_i$ and $C_i\, \neg\, C_i$
demonstrate. If a play indeed goes on
forever, then that play is won by neither of the players.

\medskip

\medskip

\noindent
Turing-machines 
exhibit precisely the kind of
behaviour captured by $\mathcal{L}$, as they can

\begin{enumerate}
\item
\emph{halt in an accepting state} 
(corresponding to Eloise---who is initially the verifier---winning the semantic game play),
\item
\emph{halt in a 
rejecting state} (corresponding to Abelard---who is the initial falsifier---winning),
\item
\emph{diverge} (corresponding to neither of the players winning).
\end{enumerate}

\medskip

\medskip

\noindent
Indeed, there  is a \emph{precise} correspondence
between $\mathcal{L}$ and Turing machines.
Let $\mathfrak{M}\models^+\varphi$ (respectively, $\mathfrak{M}\models^-\varphi$)
denote that Eloise (respectively, Abelard) 
has a winning strategy in the game beginning 
with $\mathfrak{M}$ and $\varphi$. Let $\mathrm{enc}(\mathfrak{M})$
denote the encoding of the \emph{finite} model $\mathfrak{M}$ according to some
standard encoding scheme.\footnote{The domain of a finite model can
be assumed to be a subset of $\mathbb{N}$, so an
implicit natural linear ordering is readily available for obtaining the encoding.}
Then the following theorem shows that $\mathcal{L}$
corresponds to Turing machines so that not
only acceptance and rejection but even divergence of Turing computation is captured in a
precise and natural way. The proof follows from \cite{tc}. In the theorem, by a \emph{Turing machine
for a structure problem}, we mean a Turing machine TM that gives an equivalent treatment to 
isomorphic inputs: for isomorphic $\mathfrak{M}$ and $\mathfrak{N}$, TM either 
accepts both $\mathit{enc}({\mathfrak{M}})$ and $\mathit{enc}({\mathfrak{N}})$;
rejects both; or diverges on both inputs.

\medskip

\medskip

\noindent
\begin{theorem}
For every Turing machine $\mathrm{TM}$ for a structure problem, there exists a formula $\varphi\in\mathcal{L}$
such that
\begin{enumerate}
\item
$\mathrm{TM}$ {accepts} $\mathrm{enc}(\mathfrak{M})$\hspace{1mm}
iff\hspace{3mm} $\mathfrak{M}\, {\models}^+\, \varphi,$
\item
$\mathrm{TM}$ {rejects} $\mathrm{enc}(\mathfrak{M})$\hspace{1mm}
iff\hspace{1mm} $\mathfrak{M}\, {\models^-}\, \varphi.$
\end{enumerate}
Vice versa, for every $\varphi\in\mathcal{L}$, there is a
$\mathrm{TM}$ such that the above two conditions hold.
\end{theorem}

\medskip

\medskip

\noindent
Technically this is a result in descriptive complexity theory showing that $\mathcal{L}$ captures
the complexity class RE (recursive enumerability). 
While the result concerns finite models, it is possible to
extend the result to deal with arbitrary models. The idea is to extend Turing machines to
suitable hypercomputation models while
allowing iteration in $\mathcal{L}$ to repeat for $\omega$ rounds and beyond.

\medskip

\medskip

\noindent
Since $\mathcal{L}$ captures RE, it cannot be closed 
under negation. Thus $\neg$ is not the classical negation. However, $\mathcal{L}$
has a very natural translation into natural language. The key is to
replace \emph{truth} by \emph{verification}. We read $\mathfrak{M}\models^+\varphi$ as
the claim that ``\emph{it is verifiable that $T(\varphi)$}" where $T$ is
the translation from $\mathcal{L}$ into natural language defined below.
We give two ways to translate atoms $x=y$ and $R(x_1,\dots, x_n)$. The first 
way (given in clause 1 below) covers the case where in each game position, every first-order 
variable must get a value assigned to it via the assignment function $f$.
Clause 9 gives a more careful reading for $x=y$ and $R(x_1,\dots, x_n)$ which
covers also the anomalous cases where $f$ may not 
give values to all variables.
\begin{enumerate}
\item
We translate $x=y$ and $R(x_1,\dots , x_n)$ to themselves, so 
for example $T(x=y)$ simply reads \emph{$x$ equals $y$}.
\item
The atoms $C_i$ are read as they stand, so $T(C_i) = C_i$.
\item
The $\mathrm{FO}$-quantifiers translate in the standard way, so 
 we let $T(\exists x \varphi) = \text{\emph{there exists an} $x$ 
\emph{such that} }T(\varphi)$ and analogously for $\forall x$.
\item
Also $\vee$ and $\wedge$
translate in the standard way, so
$T(\varphi\vee \psi) = T(\varphi)\text{ \emph{or} }T(\psi)$ and analogously for $\wedge$.
\item
However, $T(\neg\psi) = \text{\emph{it is falsifiable that }}T(\psi)$. Thus negation translates to
the dual of verifiability.
\item
Concerning the
insertion operators, we let $$T(\mathrm{I}x\, \varphi) =
\text{\emph{it is possible to insert a 
new element $x$ such that }}T(\varphi).$$ Similarly, we let
\begin{multline*}T(\mathrm{I}_{R(x_1,\dots , x_n)}\, \varphi) =\\
 \text{\emph{it is possible to insert a tuple }} (x_1,\dots , x_n)
\text{ \emph{into $R$ such that }}T(\varphi).
\end{multline*}
\item
Deletion operators can also be given similar natural readings.
\item
Finally, we let $$T(C_i\, \varphi\, )
= \text{\emph{it is possible to verify the claim }}  C_i 
\text{\emph{ which states that }} T(\varphi).$$
\item
We can always give the following alternative and more careful
readings to first-order atoms $x=y$ and $R(x_1,\dots , x_n)$:
\begin{enumerate}
\item
$T(x = y)$ states that \emph{the referent of $x$ is
equal to the referent of $y$}.
\item
$T(R(x_1,\dots , x_n))\ =$ \emph{the referents of $x_1,\dots, x_n$
form a tuple in $R$ in the given order.}\footnote{Note
that even ending up with an atom $x=x$ (or $\neg x = x$), without $f$
specifying a value for $x$, leads to neither player winning the game. This is
natural with the given reading for atoms. The formulae can indeed quite naturally be
considered indeterminate with respect to verification/falsification when $x$ has no value.}
\end{enumerate}
\end{enumerate}

\noindent
Thereby $\mathcal{L}$ can be seen as a \emph{simple Turing-complete fragment of
natural language}. Indeed, the simplicity of $\mathcal{L}$ is one of its main strengths.
Also, as typical computationally motivated logics translate into $\mathcal{L}$ more or
less directly, $\mathcal{L}$ can be used as a \emph{natural umbrella logic} for studying complexities of
logics. This can be advantageous, as the number of different 
logic formalisms is huge. Thus $\mathcal{L}$ offers a natural \emph{unified framework} for a programme of
studying, e.g., validity and satisfiability problems. First-order logic is \emph{not a suitable
umbrella logic} for such a programme, being expressively too weak. The expressivity of $\mathcal{L}$, on
the other hand, is of a fundamental nature, due to its Turing-completeness.
Furthermore, $\mathcal{L}$ offers a \emph{top platform} for 
descriptive complexity. Indeed, $\mathcal{L}$ can easily capture 
classes beoynd the class $\mathrm{ELEMENTARY}$, while no $k$-th order logic can. Again $\mathcal{L}$
would serve as a natural, \emph{unifying}
umbrella logic.\footnote{We note that RE, as a limit of computation, is indeed a reasonable upper 
bound for standard descriptive complexity.} All in all, $\mathcal{L}$ could be used as a
\emph{unified framework} for working on---inter alia---many kinds of reasoning issues (validity, satisfiability) as
well as topics relating to expressivity. In the next section we
analyse some further conceptual issues concerning $\mathcal{L}$.

\subsection{Further properties of $\mathcal{L}$}

\noindent 

\noindent

\noindent
It is interesting to note that $\neg$ can be 
read as the classical negation (rather than falsifiability) in those fragments of $\mathcal{L}$ where the 
semantic games are determined. Standard FO is such a fragment.
Furthermore, adding a generalized quantifier to $\mathcal{L}$ 
corresponds to adding a corresponding oracle to Turing machines; 
see \cite{tc} for further details.

\medskip

\medskip

\noindent
In our system, under our formal and fully explained
semantics, the sentence $C_i\neg C_i$ is indeterminate, and so is $C_i C_i$. This should be natural from any
perspective that accepts the semantics we gave.
To further analyse whether this is natural, let us consider $C_i\, C_i$ first. Now, a typical logic (such as FO) is
compositional, with well-founded formulae. This means each formula is essentially an
algebra term $f(t_1,\dots , t_k)$.
And the formula $f(t_1,\dots , t_k)$ has a \emph{meaning} which is determined by
applying the function $f$ to the \emph{meanings} of $t_1,\dots, t_k$.
The well-foundedness means that the algebra term is finite, and ultimately has \emph{atomic}
formulae $x_1,\dots , x_k$ whose \emph{meaning} is fully determined in some uncontroversial and
independent way. Thus we can evaluate $f(t_1,\dots , t_k)$ in a
finite process, since the ultimately reachable atoms have already fully defined, independent meanings. Such logical
reductionism is handy indeed.

\medskip

\medskip

\noindent
However, at least in the sense of our semantics, $C_i C_i$ does not
have this kind of a well-founded evaluation process. Syntactically $C_i C_i$ is an algebraic term (the first $C_i$ is an
operator and the second one an atom). However, semantically, the meaning of the atom $C_i$ is not
already defined, but instead, it must be evaluated based on the full formula $C_i C_i$ (because from the atom $C_i$ we
jump back to the operator $C_i$ and continue checking from there). Therefore the meaning of $C_iC_i$ is
defined based on $C_i C_i$ itself. Thus it is natural to consider the sentence indeterminate.
The same holds for $C_i \neg C_i$. It
also tries to define its meaning based on itself. It is indeed natural to require meanings to be dug from an
external source in a reductionist way: if we define $q$ to be true if and only if $q$ is true, and no further
information about the situation can appear, it is
natural to consider $q$ indeterminate. Digging up the truth value
from the atomic level is impossible in the case of $C_i C_i$ and $C_i\neg C_i$. We note that, \emph{if}
one accepts the semantic game of $C_i\neg C_i$ to also be the
evaluation procedure of the actual liar sentence, then the 
explained lack of well-foundedness applies as such. This leads to an indeterminate truth value. However, of course, 
this can be considered paradoxical, as now the statement ``this sentence is false'' seems false, as the
sentence was supposed to be
indeterminate. But false is not
indeterminate, and so onwards, in the usual way, it seems to get different flipping truth values.\footnote{In
our system with our reading, $C_i\neg C_i$ states roughly that ``this sentence is falsifiable.'' 
The sentence formally evaluates to indeterminate. Now one could assert that since indeterminacy does not
equal falsifiability, we can conclude that the sentence is in fact false (meaning that
\emph{falsifiability was not the case}).
Now, ``false'' here meant
``\emph{falsifiability was not the case}'' and
\emph{indeterminate} is therefore consistent with this meaning of false.
(Note that we do not have ``false'' in our system.
Simply falsifiable, verifiable and indeterminate. We have dictated a semantics, we do not
compare truth values to anything external to the semantics.)}

\medskip

\medskip

\noindent
As we have discussed above, our formal systems, as 
defined in Section \ref{systems}, can be used to model a wide
variety of dynamical frameworks rather naturally. 
Now, it is obvious that the \emph{semantic
games of the logic $\mathcal{L}$ are systems} in our formal sense. Thus $\mathcal{L}$
can be directly used, inter alia, to model evolving physical frameworks. However, $\mathcal{L}$ is also a
natural setting for formalizing mathematics.
Indeed, $\mathcal{L}$ can be used as a possible, highly strict measure of
what counts as a \emph{mathematical claim}. Indeed, mathematics is 
\emph{intuitively and informally} something fully rigorous and somehow 
predetermined and objective. It is often \emph{considered} somehow mind independent and perhaps even of a
Platonic nature. Now, can we capture this intuition of strict objectivity?

\medskip

\medskip

\noindent
A nice starting point for capturing the
intuition would be to assert that a claim is
mathematical if we can determine whether it holds using some uniform
systematic procedure. The idea here is that there exists a systematic 
procedure $\mathbb{P}$ with a carefully defined set of inputs and the set $\{yes, no\}$
of outputs. The (not necessarily nice) requirement here is that the set of inputs $\mathbb{I}$ is somehow
rigorously fixed and quite limited. A natural option here would be that the
set $\mathbb{I}$ must be somehow extremely simple (for
example the collection of all finite strings over the alphabet $\{0,1\}$ or
the---suitably simple and certainly
decidable---collection of 
all formulae of some logic; we are thinking about $\mathcal{L}$ here).

\medskip

\medskip

\noindent
Another (not necessarily nice) requirement is that we
must pick a \emph{single} systematic procedure $\mathbb{P}$ to check, for 
each input $i\in\mathbb{I}$, whether $i$
holds or not.\footnote{A Turing machine is a systematic procedure. Sometimes 
Turing machines are described to capture what can be \emph{mechanically} 
executed. But ``mechanical'' is perhaps not as good a word as ``systematic.'' This is because the word
``mechanical'' has a quite strong connotation relating to physicality. Physical
systems can do things that seem more or less impossible to explain/calculate/describe, even in
principle. Turing machines are 
physically realizable in principle, but the converse (from physically realized 
devices to Turing machines) is problematic. This is because we cannot tell
precisely what the full components of an actual, physically realized device are.
It can be more or less impossible to somehow 
write down a precise Turing specification based on the physical construct. For example, \emph{in principle}, a
series of coin tosses could keep giving heads on precisely the rounds $j\in S\subseteq \mathbb{N}$,
where $S$ is undecidable.  Given a physically realized device, perhaps it is essentially a Turing machine, but the 
problem is that it is hard to know which one. Thus it may be more to the point to make the
hypothesis that Turing machines 
capture the notion of \emph{systematic} executability rather than mechanical. Nevertheless, it can of
course even be natural to make the \emph{hypothesis} that nature is a
essentially a Turing machine, but this does not imply that we understand, simply by looking at
physcal systems, what machine that systems should correspond to. This is especially true if we
cannot---and it seems we cannot---fully isolate the system from its environment.}
Now, $\mathbb{I}$ is precisely the set of mathematical statements, and we have a systematic and 
somehow objective procedure $\mathbb{P}$ for verifying truth\footnote{Indeed, one could claim
that $\mathbb{P}$ even \emph{defines} (or can be considered to define) which claims hold. And we will indeed 
consider the scenario where $\mathbb{P}$ determines truth.} of
the statements, but $\mathbb{P}$ does not have to
produce an output on every input, so $\mathbb{P}$ could
correspond to a Turing machine. \emph{Thus it is possible to consider
formulae of $\mathcal{L}$ to be $\mathbb{I}$}. 
For each input $\varphi\in\mathcal{L}$, we
check whether $\varphi$ is verified or falsified in
the empty model.\footnote{More rigorously, we consider the empty model in
the signature of $\varphi$.} It is of course possible that $\varphi$ is neither verified nor falsified.
The way the procedure $\mathbb{P}$ now works is, for its essential parts,
described in \cite{tc}, proof of Theorem 4.3.
The nice thing is that $\mathcal{L}$ is a
logic, so our inputs are statements rather than, e.g.,
binary strings.\footnote{The setting is, however, reasonably similar to
equating mathematical statements 
with Turing machines with the empty input.}

\medskip

\medskip

\noindent
Now, the setting is still quite restrictive, because $\mathcal{L}$ does not directly talk about, e.g., infinite sets.
Thus it will not be sensible to equate the setting with real mathematics.
But it \emph{can} be viewed as a possible formulation of the \emph{strict core} of
mathematics. The framework based on $\mathcal{L}$ is objective, finitary and \emph{captures} the notion of
systematicity. It indeed fully and formally
captures systematicity if we define systematicity according to the Church-Turing
thesis to correspond to Turing machines.
In a sense, systematicity is also precisely and exactly what \emph{logic} is all about, so it is possible to
entertain the view that $\mathcal{L}$ provides a definition of logicality.\footnote{We do not really
differentiate between logicality and mathematicality here, but instead identify both notions 
with the notion of rigorous objective systematicity.}

\medskip

\medskip

\noindent
Note especially that the perspective of using $\mathcal{L}$ to define 
mathematical statements banishes typical incompleteness issues. Every statement $\varphi\in\mathcal{L}$
corresponds to posing the questions  ``$\emptyset\models^+ \varphi$ ?'' and ``$\emptyset\models^- \varphi$ ?''.
The answer is given by $\mathbb{P}$.
If $\emptyset\models^+ \varphi$, then $\mathbb{P}$ outputs \emph{yes}, and if $\emptyset\models^- \varphi$,
then $\mathbb{P}$ outputs $\emph{no}$. If $\mathbb{P}$
diverges,\footnote{Here divergence is equated 
with $\emptyset \not\models^+\varphi$ and $\emptyset \not\models^-\varphi$, i.e., neither player
having a winning strategy.}
then $\mathbb{P}$ will not output anything. Indeed, we take $\mathbb{P}$ to 
\emph{define} truth and falsity here. Thus there are no true but not verifiable (or false but unfalsiable) staments.
The system \emph{defines} truth values, nothing external does. (It is worth noting that the statements 
where $\mathbb{P}$ diverges are simply considered not to have any truth value---unless 
divergence/indeterminacy is a truth value, which is then the truth value of those sentences.)

\medskip

\medskip

\noindent
In conclusion, $\mathcal{L}$ gives a possible, strict standard for strict mathematicity (or logicality, 
we do not differentiate here). The thesis that logicality equals systematicity
can be appealing, and if systematicity equals Turing executability (RE), then $\mathcal{L}$ hits some
fundamental mark. Thus it could be regarded as a fundamental logic. For those in favour of the
perspective that there is a 
unique fundamental logic, $\mathcal{L}$ could perhaps be one possible candidate. But of course this
requires one to favour (1) the uniqueness thesis; (2) the idea of logicality being systematicity; (3) 
the Church-Turing thesis that systematicity is captured by Turing machines;
and (4) the position that $\mathcal{L}$ should be a system
capturing Turing computation in some fundamentally natural way. The naturality could be
due to the links between $\mathcal{L}$ 
and natural language and the apparent minimality of $\mathcal{L}$ in achieving its central 
features such as the structure modification capacities, self-reference, and the containment of FO.
This can be quite a lot to entertain, but seems interesting nevertheless.

\subsection{Controlling systems with $\mathcal{L}$ and its variants}

\noindent
We have already noticed that the semantic games of $\mathcal{L}$ are system evolutions. This is
easy to see---and the system related to one evaluation task (i.e.,
checking, e.g., if $\mathfrak{M}\models^+\varphi$)
can be realized in many ways.
Roughly, we can take Eloise to be the sole agent and associate Abelard with $G$.\footnote{It 
can also be natural to associate different
operators ($\exists x_1$, $\exists x_2$, $\mathrm{I}_{R(x_1,\dots , x_n)}$ etc.)
with different agents. Here one may even want to take into account the polarity of the operator in
each turn, i.e., whether Eloise or Abelard uses the operator, and associate different 
modes of use with different agents. Also, a different occurrence (token) of
the same operator in the same formula can be associated with a different agent.}
The system constraint function $F$ relates to the constraints
given by the formula evaluated (and the semantic game rules together with the input models).
This is a turn-based game, so we need dummy moves. The current world is the current
model assignment pair, and we can put the remaining situation specification (where the
game is in relation to the evaluated
formula, and who is the current verifier) into mental models.\footnote{They can be made part of
the current model as well.}

\medskip

\medskip

\noindent
This phenomenon has a converse. A typical system controllable by Turing machines
can reasonably naturally be simulated relatively closely in a
setting with a semantic game of $\mathcal{L}$. The key feature of $\mathcal{L}$ is
the looping operator, allowing semantic games with indefinitely long plays. The other 
key feature is the possibility to modify models and thereby simulate phenomena relating to the structure 
evolution in the system modeled. Typically Eloise controls
the main agents for most parts. The choices of the agents can be represented by direct 
modifications of the current model. If necessary, we can add a some novel
points and a predicate $A$ that those points satisfy, and then choices can be represented by, e.g.,
colouring the elements in $A$ with different singleton predicates. \emph{Intuitively}, Abelard controls $G$ and 
modifies the structures so that the new current models become as desired. However, many kinds of 
modeling solutions can be made, based on what kinds of correspondences between the original game and 
the semantic game are desired. Note that there is no explicit winning notion in any way present in 
systems, while semantic games are reachability games. Nevertheless, Eloise and Abelard are free to make any
choices within the rules of the semantic game framework, and indeed, they are not
obliged to try to win the game plays.

\medskip

\medskip

\noindent
To allow more flexibility in model constructions with $\mathcal{L}$, we consider an
extension $\mathcal{L}[\, ;]$. The language is obtained by 
extending that of $\mathcal{L}$ by 
the formula construction rule stating that if $\psi$ and $\varphi$ are
formulae, then so is $\psi\, ;\varphi$. Intuitively, `$;$' is a composition operator stating that the 
left formula must first be checked and then the right one, but `$;$' also intuitively
corresponds to a conjunction. Formally, if we encounter $\psi\, ; \varphi$,
we first play the game for $\psi$ (with the model and assignment as they are when
the formula $\psi\, ;\varphi$ is encountered). The current verifier $P$ (Eloise or Abelard) of $\psi\, ; \varphi$
has the task to win the play of the game for $\psi$. If that play
for $\psi$ ends so that $P$ wins, the game play is
continued from $\varphi$ with the possibly
modified model and assignment\footnote{So the game for $\psi$ can of
course modify the model and assignment, and the game for $\varphi$ is 
then begun with these modified variants.} and with $P$ being the verifier for $\varphi$ (despite the 
possible changes of verifier in the play for $\psi$). 
If the play for $\psi$
does not end, then neither 
player wins the entire original play, and if the opponent of $P$ wins the play for $\psi$, then the 
original play ends with the opponent winning. Note that $\psi$ may have further symbols `$;$' and
jump symbols $C_i$ (so we may even, e.g., encounter the very same
formula $\psi\, ;\, \varphi$ again due to a jump). Each 
encounter of `$;$' has to be resolved first before earlier encounters (that have 
occurred earlier in the run of
the game play) will be returned to for checking the second formula $\varphi$.\footnote{But
jumps can of course cause all kinds of intermediate actions including possibly playing inside the 
syntactic structure of the second formula.} The full game
history can be used here to define strategies; in
standard $\mathcal{L}$ a positional 
history is the standard option. Nondeterministic strategies can be allowed, being winning if
every path leads to a win.

\medskip

\medskip

\noindent
This composition operator `$;$' 
does not affect the Turing-completeness, as the old proof applies for
translating from logic into Turing computation (and the direction from Turing
computation to logic does not
need the composition operator in the first place). The extended logic offers more
flexibility in simulating model constructions that Turing machines do. The key issue intuitively is
that we can write formulae where Eloise does not have to stop working on a model simply because
Abelalard doubts some true atom, thus forcing Eloise to win before the model
construction has a desired form. A result concerning model constructions follows.
Let us formulate that next.

\medskip

\medskip

\noindent
Let us write $\mathfrak{M}\models^+ (\varphi,\mathcal{M})$ if Eloise has a winning
strategy in the
game for $\varphi\in\mathcal{L}[\, ; ]$ and $\mathfrak{M}$ so that the set of
models where the plays end with that strategy is precisely $\mathcal{M}$. Note that different 
strategies can produce different sets.
Here we ignore all the symbols outside the signature (first-order variables and the relation symbols $X$
corresponding to tape symbols) so the signature of models in $\mathcal{M}$ is the same as
that of $\mathfrak{M}$. There may be other symbols present (e.g., $X$) in the final model constructions, but
they do not count.

\medskip

\medskip

\noindent
We say that a nondeterministic TM \emph{computes a $\sigma$-structure relation} if for
every finite $\sigma$-structure encoding $\mathit{enc}(\mathfrak{M})$ given as input, the
machine always halts (on every computation path)
and the collection of the outputs is $\mathit{enc}(\mathfrak{N}_1),\dots ,
\mathit{enc}(\mathfrak{N}_k)$ for some $\sigma$-structures $\mathfrak{N}_1,\dots ,
\mathfrak{N}_k.$ Furthermore, for every $\mathfrak{M}'$
isomorphic to $\mathfrak{M}$, the output collections with inputs $\mathit{enc}(\mathfrak{M})$
and $\mathit{enc}(\mathfrak{M}')$ contain representatives of precisely the same 
isomorphism classes.

\medskip

\medskip

\noindent
Now, for any TM computing a $\sigma$-structure relation,
there exists a formula $\varphi$ of $\mathcal{L}[\, ; ]$ such that the following
conditions both hold for all finite $\sigma$-structures $\mathfrak{M}$.

\begin{enumerate}
\item
TM produces, on the input $\mathit{enc}(\mathfrak{M})$, encodings of models that define
precisely the collection $\mathcal{K}_\mathfrak{M}$ of isomorphism classes.
\item
$\mathfrak{M}\models^+ (\varphi,\mathcal{M}_\mathfrak{M})$ with $\mathcal{M}_\mathfrak{M}$ defining
the collection $\mathcal{K}_\mathfrak{M}$ of isomorphism classes.
\end{enumerate}

\noindent
The way to establish this is to define a formula $\varphi$ that forces the
following construction. Beginning with the input model $\mathfrak{M}$, 
the formula $\varphi$
forces Eloise to construct the computation table with the 
input $\mathit{enc}(\mathfrak{M})$. The newly constructed model $\mathfrak{S}$
has the input model $\mathfrak{M}$,
the computation table consisting of word models, and one chosen
output model $\mathfrak{N}$. (The construction
allows all of the desired output models to be chosen, one output model at a time.) The 
computation table as well as  $\mathfrak{M}$ and the output model are coupled with suitable
auxiliary relations so that, e.g., the first and last bit strings in one branch of the
computation table represent the word 
models corresponding to $\mathit{enc}(\mathfrak{M})$ and $\mathit{enc}(\mathfrak{N})$.
Also other auxiliary constructs can be added, e.g., to distinguish the original elements of the input 
model from the newly added domain elements.
The computation table and the 
auxiliary relations can be constructed from
tape symbol predicates, but the output model $\mathfrak{N}$ 
consists of signature relations.
Once this construction $\mathfrak{S}$ is ready, a first-order formula $\psi$ is
evaluated. The first-order formula asserts that the computation table and the whole
construction is as desired.\footnote{Note that the empty set of output models does not feature here, and at 
least the empty model is outputted on the logic side anyway. Nevertheless, this could be remedied with
suitable and reasonably natural constructions.} Then the undesired part (i.e., everything but $\mathfrak{N}$) is
deleted and finally Eloise wins. All this can be done using the composition operator which prevents Abelard 
stopping Eloise before suitable construction and deletion parts are done.
Using the auxiliary tape symbol predicates helps to distinguish between different construction parts at
different stages of the full procedure.

\medskip

\medskip

\noindent
We note that of course the above construction generally requires looping.
And the construction remotely resembles the way the logic $\mathcal{L}_{\mathit{RE}}$ of \cite{okts}
achieves Turing-completeness; cf. Proposition 2 there. The logic $\mathcal{L}_{\mathit{RE}}$ is
based on an opertor $\mathrm{I} Y$ that adds a finite set of points to the model and labels 
the new points with the unary predicate $Y$.

\medskip

\medskip

\noindent
It is actually interesting to consider further operators and
extensions of that framework.
Firstly, one can allow (1) \emph{free use} of operators $\mathrm{I} Y$, rather than the
prenex restricted single operator allowed in $\mathcal{L}_{\mathit{RE}}$.
Furthermore, one can consider (2) an operator $D$ that simply deletes some  number of points from the 
model domain; (3) operators $(\mathrm{I} R)$ that introduce an arbitrary set of tuples to the relation $R$ within
the current model;
and (4) an operator $(\mathrm{D} R)$ that deletes an arbitrary set of tuples from the relation $R$.
We conjecture that
adding all these to second-order logic leads to a system that corresponds to the arithmetic hierarchy (in
expressivity) in the finite. Of course one needs to deal again with atoms with
first-order variables without a reference. A natural
approach here is to define all such atoms false, so atoms would basically state the additional
condition that the symbols in the atom have referents, resembling the Russelian approach.
This strong but seemingly simple system is very flexible in relation to model constructions, and
especially its fragments are immediately interesting.\footnote{We note that of course we can also
consider $\hat{D}Y$
which simply deletes the points in the extension of the unary relation $Y$ and an operator $I$ that simply
adds new domain points to the model without labeling them. Modifications can also be
restricted to the extension of an input formula. This defines variants of the operators we listed,
for example $\langle\mathrm{I}(\varphi(x,y),R)\rangle\chi$ would add all pairs that
satisfy $\varphi(x,y)$ to the binary 
relation $R$, and a similar deletion operation of course would also be natural. 
Concerning domain points, the formula $\langle D(\psi(x))\rangle \chi$ would
delete the points that satisfy $\psi(x)$. And the list goes on. After
the modification in each case, $\chi$ would be evaluated. To go more directly beyond first-order logic, one could 
consider variants where only some
nondeterministically chosen set of
tuples/points in the extensions could be added/deleted.}


\medskip

\medskip

\noindent
To define more custom-made logics than the very general logics
based on $\mathcal{L}$, let us turn to
fragments and further variants. To investigate model
transformations, let us define \emph{modifiers}. These can be used for jumps from models 
to other models, similarly to what happens in $\mathcal{L}$.
Modifiers are defined as follows.  Let $S$ be a class of pairs $(\mathfrak{M},X)$
where $\mathfrak{M}$ is a first-order structure and $X$ an assignment; $X$
can also be a team or a domain point,
depending on the exact application. Futhermore, to streamline our exposition, $(\mathfrak{M},X)$ can even 
represent a class of structures $(\mathfrak{N},f)$ where $\mathfrak{N}$ is a first-order
model and $f$ an assignment. Such classes (called model sets) are 
considered in \cite{a17k}.
Now, fix one of the above possible interpretations for structures $(\mathfrak{M},X)$.

\medskip

\medskip

\noindent
A modifier $m$ is a
map $$m: S\rightarrow \mathcal{P}(S)$$
such that if $(\mathfrak{M},X)\cong (\mathfrak{N},Y)$ for some two elements of $S$,
then there is a bijective map $p: m((\mathfrak{M},X))\rightarrow m((\mathfrak{N},Y))$
such that $(\mathfrak{A},U)\cong p((\mathfrak{A},U))$ for all inputs $(\mathfrak{A},U)$ to $p$.
Now, mixing syntax and
semantics, $(\mathfrak{M},X)\models (m)\varphi$ iff $(\mathfrak{N},Y)\models\varphi$ for 
all $(\mathfrak{N},Y)\in m((\mathfrak{M},X))$.

\medskip

\medskip

\noindent
Note that if $U$ and $V$ are teams, then $(\mathfrak{A},U)\cong (\mathfrak{B},V)$
if $(\mathfrak{A},\mathit{rel}(U))\cong(\mathfrak{B},\mathit{rel}(V))$
and $U$ and $V$ have the same domain. The relations of teams are determined in the usual way, using the ordering of the
subindices of the variable symbols to determine the internal ordering of tuples. If $U$ and $V$ are 
assignments, then they correspond to singleton teams, so the above specification suffices to define $\cong$.
If $U$ and $V$ are domain points, they correspond to singleton assigments, so again the case is covered.
In the case of model sets $T$ and $T'$, we have an isomorphism if there is a bijective map $g$ from $T$ to $T'$
such that $(\mathfrak{N},f)$ is isomophic to $g((\mathfrak{N},f))$ for
all inputs $(\mathfrak{N},f)$ to $g$.

\medskip

\medskip

\noindent
Modifiers are a simple way to jump from structures to others.
Altering things a bit, define $(\mathfrak{M},X)\models \langle m\rangle \varphi$
iff $(\mathfrak{N},Y)\models\varphi$ for 
some $(\mathfrak{N},Y)\in m((\mathfrak{M},X))$.
One can also consider  variants with, e.g., ``most.''  Going further, fix some
Boolean formula $B$ with $k$ proposition symbols.  Define
$(\mathfrak{M},X)\models (( F )) (\varphi_1,\dots, \varphi_k)$
if and only if the Boolean combination $B$ of
the statements $(\mathfrak{N}_i,Y_i)\models\varphi_i$ holds
for each tuple $$((\mathfrak{N}_1,Y_1),\dots , (\mathfrak{N}_k,Y_k))\
\in F((\mathfrak{M},X)),$$

\medskip

\medskip

\noindent
where $F$ maps from $S$ into $\mathcal{P}(S^k)$.
Here, if $(\mathfrak{M},X)\cong(\mathfrak{M}',X')$,
then there is a bijection $h:F((\mathfrak{M},X))\rightarrow F((\mathfrak{M}',X'))$
such that the $j$th entry of $((\mathfrak{N}_1,Y_1),\dots , (\mathfrak{N}_k,Y_k))$ is
isomorphic to the $j$th entry of $$h(((\mathfrak{N}_1,Y_1),\dots , (\mathfrak{N}_k,Y_k)))$$
(for all $j$ and all inputs to $h$).

\medskip

\medskip

\noindent
This last notion of a modifier is quite general, but it still modifies models 
independently of the input formulae $\varphi_1,\dots , \varphi_k$. And there are other
limitations.  Thus let us define a somewhat more general notion. Let $G$ be a
function from $S\times (\mathcal{P}(S))^k$ into $\mathcal{P}((\mathcal{P}(S))^{2k})$.
Intuitively, $G$ takes as input the current structure $(\mathfrak{M},X)$ and
the truth classes $$\parallel \varphi_i \parallel
\ =\ \{ (\mathfrak{N},Z) \in S\, |\, (\mathfrak{N},Z)\models \varphi_i\}.$$
The output of $G$, then, is a collection of $2k$-tuples $(\mathcal{R}_1^+,\mathcal{R}_1^-,\dots ,
\mathcal{R}_{k}^+, \mathcal{R}_{k}^-)$ of
classes $$\mathcal{R}_1^+,\mathcal{R}_1^-,\dots ,
\mathcal{R}_{k}^+, \mathcal{R}_{k}^- \subseteq S.$$ 
There should be at least one tuple so that
each model in $\mathcal{R}_{i}^+$ satisfies $\varphi_i$ while no
model in $ \mathcal{R}_{i}^-$ satisfies $\varphi_i$. Intuitively, each tuple approximates a
\emph{type} consisting of simple satisfaction and negative satisfaction statements. Thus the
collection that $G$ outputs corresponds to a \emph{disjunction of types} (or
type approximations).
So, we indeed define $(\mathfrak{M},X)\models (( G )) (\varphi_1,\dots, \varphi_k)$
if and only if for some $t := (\mathcal{R}_1^+,\mathcal{R}_1^-,\dots ,
\mathcal{R}_{k}^+, \mathcal{R}_{k}^-)$ in $$G((\mathfrak{M},X),\parallel
\varphi_1\parallel, \dots , \parallel \varphi_k \parallel),$$
we have $(\mathfrak{M}_i,Y_i)\models\varphi_i$ for all $(\mathfrak{M}_i,Y_i)\in\mathcal{R}_i^+$
and $(\mathfrak{M}_i',Y_i')\not\models\varphi_i$ for all $(\mathfrak{M}_i',Y_i')\in\mathcal{R}_i^-$
(for all $i\leq k$ obviously).

\medskip

\medskip

\noindent
Again of course some invariance conditions hold for $G$.
We require that if $(\mathfrak{M},X)\cong(\mathfrak{M}',X')$,
then there exists a bijection
$$h:G((\mathfrak{M},X),\parallel \varphi_1\parallel, \dots , \parallel \varphi_k \parallel)
\rightarrow G((\mathfrak{M}',X'),
\parallel \varphi_1\parallel, \dots , \parallel \varphi_k \parallel)$$
such that for all tuples $$t\in G\bigl((\mathfrak{M},X),\parallel
\varphi_1\parallel, \dots , \parallel \varphi_k \parallel \bigr),$$
and all $j\leq 2k$, there exists a bijection $p$ from the $j$th entry of $t$ to
the $j$th entry of $h(t)$ so that each model $(\mathfrak{N},Y)$ in the domain of $p$ is
isomorphic to $p((\mathfrak{N},Y))$.

\medskip

\medskip

\noindent
Such modifiers $G$ are
reasonably general, so we call such modifiers \emph{g-modifiers}, g for general.
Other kinds of invariance conditions could be natural, but we shall not address 
that issue here further than a short comment later on below.

\medskip

\medskip

\noindent
We consider one more notion of very closely related operators.
We define \emph{operators} $N$
such that $$(\mathfrak{M},X)\models \langle\langle H\rangle\rangle (\varphi_1,\dots, \varphi_k)$$ if
and only if $\bigl( \parallel
\varphi_1\parallel, \dots , \parallel \varphi_k \parallel \bigr)\ \in\ N\bigl( (\mathfrak{M},X) \bigr)$.
Here $N$ maps from $S$ into $\mathcal{P}((\mathcal{P}(S))^k)$. A natural invariance 
here is that if $(\mathfrak{M},X)$ and $(\mathfrak{M}',X')$ are isomorphic,
then there is a bijection $h$ from $N(\mathfrak{M},X)$ to $N(\mathfrak{M}',X')$
such that for all $t\in N((\mathfrak{M},X)),$
the $j$th entry of $t$ and the $j$th
entry of $h(t)$ are related by a bijection $p$ such 
that the structures linked via $p$ are isomorphic. This resembles neighbourhood semantics. We note
that it is of course not necessary to impose the above invariance conditions if not desired, but
that may result in isomorphic structures $(\mathfrak{M},X)$ and $(\mathfrak{M}',X')$
being non-equivalent.

\medskip

\medskip

\noindent
We note that the above classes of logical constructs (i.e., operators and modifiers)  can be
equivalently formulated in terms of functions on isomorphism classes.
We say that $F$ and $G$ are \emph{i-similar} if they are the same modulo modifying the involved
structures to 
isomorphic ones. 
It is also worth noting that natural construct classes arise, e.g., by considering
which constructs are the same modulo permutation of
variables (i.e., a bijection from 
the set VAR of all variable symbols to VAR itself). $F$ and $G$ are \emph{x-similar} if they are the same
modulo variable permutations. $F$ and $G$ are \emph{similar} if they are the same
modulo modifying the involved structures to 
isomorphic ones and permuting variables. All the three notions of similarity define relevant 
classes of logical constructs that are, in some senses, the same, or have the same base construct.

\medskip

\medskip

\noindent
Classifying modifiers via $\mathcal{L}$ and $\mathcal{L}[\, ; ]$ can be
elucidating and useful, as those logics 
modify models one tuple and one domain point at a time. Already the self-reference-free
fragment of $\mathcal{L}$ defines a natural and highly comprehensive class of simple, finite step
modification procedures. 
Nevertheless, modifiers can be natural in many contexts.
For example, for more flexible studies of formal systems (as
defined above), it would be useful to investigate extensions of $\mathcal{L}$ and its
variants with more
custom-made model modification steps.

\medskip

\medskip

\noindent
Let us then move on to briefly discuss model sets. The discussion will relate to
knowledge representation and the issues in Section \ref{falsepartial}. 
Suppose we consider systems
where the mental model is the conceived set of possible current models.
This is a classical approach. Now, for those systems, we 
can directly use model sets \cite{a17k}. 
Using the team semantics of \cite{a17k} on a first-order
formula $\varphi$, we have $$\mathcal{M}\models \varphi
\text{ if and only if }\mathfrak{M},f\models_{\mathrm{FO}}\varphi
\text{ for all }(\mathfrak{M},f)\in \mathcal{M},$$
where $\mathcal{M}$ is a model set, i.e., a collection of pairs $(\mathfrak{M},f)$
where $f$ is an assignment. As established in \cite{a17k}, 
this variant of team semantics gives essentially a semantics for proofs. Disjunction 
corresponds to splitting into cases and negation to going from verification to falsification.
In Section \ref{falsepartial} we discussed the possibility of using truncated 
reasoning when determining the output of the decision function $d_i$. A semantics for proofs can be 
useful in this context, as typical proof steps---such as splitting into cases---are
reflected in the semantics. It is
interesting, e.g., to consider what can be established with a strongly limited number of such
semantic counterparts of proof steps.
All this directly relates to issues in knowledge representation. 
Indeed, consider querying under the open world setting. It is all about \emph{dealing with very delicate
consequence relations}. Let us see an example of open world querying and relate it to model sets.

\medskip

\medskip

\noindent
Let $(\sigma,\mathcal{O},q(\overline{x}))$ be an ontology-mediated query (see \cite{obda}).
Here we define $\mathcal{O}$ to be an ontology (possibly in some strong logic), $\sigma$ a
signature and $q(\overline{x})$ a
query over $\sigma\cup\mathit{signature}(\mathcal{O})$.
Let $\mathcal{F}$ be a $\sigma$-database, i.e., a
set of literals\footnote{Here we allow positive \emph{and}
negative relational facts} over the signature $\sigma$.
Let $\overline{a}$ be a tuple of elements occurring in $\mathcal{F}$.
We here define that $\mathcal{F} \models ( \sigma, \mathcal{O},  q(\overline{a}))$ if
and only if $\mathfrak{M}\models q(\overline{a})$ for all
models $\mathfrak{M}$ of the
signature $\sigma\cup\mathit{signature}(\mathcal{O})$
such that

\begin{enumerate}
\item
$\mathfrak{M}\models\bigwedge \mathcal{O}$
\item
The diagram of $\mathfrak{M}$ contains $\mathcal{F}$ as a subset.
\end{enumerate}

\noindent
Let $\mathcal{M}[\sigma,\mathcal{O},\mathcal{F},\overline{x}\mapsto\overline{a}]$
denote this
model set (defined by the above two conditions),
with every assignment mapping the elements of $\overline{x}$ to
the respective elements of $\overline{a}$. Then we have $\mathcal{M}[\sigma,\mathcal{O},\mathcal{F},
\overline{x}\mapsto\overline{a}]\models q(\overline{x})$ if
and only if $\mathcal{F} \models ( \mathcal{O}, \sigma, q(\overline{a}))$.
Thus we can turn the logical consequence issue into model set satisfaction, which uses
the team semantics of model sets. (This obvious connection of ontology-based
data access to model sets was briefly noted in \cite{a19k}.)
As already discussed, all this can be useful when considering different  
reasoning notions with limited reasoning capacities and truncated reasoning patterns.
Note that, if desired, we can put even quite severe cardinality limits to 
the models in the model set. And we can stay in the finite if we want to. If we want more 
complex data than literals, an approach via model sets can still be used. It is simply about
delicate consequence relations, and model sets relate directly to those.

\medskip

\medskip

\noindent
To seriously study delicate consequence relations used in knowledge representation, one
must understand very delicate fragments of FO and $\mathcal{L}$, as this 
helps in various kinds of classification attempts. For example,
antecedent formulae could be only atoms, while consequent formulae are more elaborate. For 
such studies, we need tools for flexible, fine-grained classification.
To classify logics in a flexible, delicate and 
\emph{very fine-grained way}, it would be beneficial to have access to related 
algebraic approaches. These are not difficult to obtain.
In the next section we 
take some related first steps.

\subsubsection{First-order logic and extensions via functions on relations}\label{randomsection}

Here we define an algebraic approach to first-order logic. The system resembles the 
approach of Codd, but employs a finite signature and considers 
standard first-order logic It is also very close to predicate functor logic, but we
also consider some extensions of the system below.

\medskip

\medskip

\noindent
The key is to deal with identities and 
relation permutations by operations that operate only in
the beginning of tuples. We can arbitrarily permute any tuple by combining
swaps of the first two coordinates with a
cyclic permutation operation. Furthermore, we can identify (i.e., force equal) any two tuple
elements by first bringing the elements to the beginning of a tuple and then applying an identity 
operation that checks only the first two coordinates.
What we formally mean by these claims will be of course made clear below.

\medskip

\medskip

\noindent
A first-order formula $\varphi(x_1,\dots , x_k)$ defines the 
relation $$\{(a_1,\dots, a_k)\in A^k\, |\, \mathfrak{A}\models\varphi(a_1,\dots , a_k)\, \}$$
over the model $\mathfrak{A}$ (where $A$ is the domain of $\mathfrak{A}$).
This requires that the free variable symbols $x_i$ 
are linearly ordered. Here we let the linear ordering be associated with the 
subindices of the variable symbols. Now, what would be the relation 
defined by the formula $\varphi(x_2,\dots , x_{k+1})$? It would be 
natural to let it be $$\{(a_2,\dots, a_{k+1})\in A^k\, |\, \mathfrak{A}\models\varphi(a_2,\dots , a_{k+1})\, \}.$$
This is precisely the same relation as the relation 
given by $\varphi(x_1,\dots, x_k)$ because we obviously have
$$\{(a_2,\dots, a_{k+1})\, |\, \mathfrak{A}\models\varphi(a_2,\dots , a_{k+1})\, \}
= \{(a_1,\dots, a_{k})\, |\, \mathfrak{A}\models\varphi(a_1,\dots , a_{k})\, \}.$$
One way around this is to let formulae define sets of
assignment functions, i.e., the ``relation'' defined by $\varphi(x_1,\dots , x_k)$
over $\mathfrak{A}$ is now $$\{\bigl((x_1,a_1),\dots, (x_k,a_{k})\bigr)\, |\, 
\mathfrak{A}\models\varphi(a_1,\dots , a_{k})\, \}.$$
And the ``relation'' defined by $\varphi(x_2,\dots , x_{k+1})$
over $\mathfrak{A}$ is $$\{\bigl((x_2,a_2),\dots, (x_{k+1},a_{k+1})\bigr)\, |\, 
\mathfrak{A}\models\varphi(a_2,\dots , a_{k+1})\, \}.$$
So, in some sense, first-order formulae do not really define relations over
models but sets of assignments (which could be characterized as \emph{index labeled 
relations}.)\footnote{It is worth noting that relational database theory is not based on relations but
these kinds of labeled relations.}

\medskip

\medskip

\noindent
Now, we shall here work with relations, not sets of assignments. The relation 
defined by a first-order formula $\varphi$ in a model $\mathfrak{A}$ is, strictly speaking,
specified as follows.
\begin{enumerate}
\item
Let $(x_{i_1},\dots, x_{i_k})$ enumerate exactly all the free variables in $\varphi$, with
the subindices $i_1,\dots , i_k$ given in a strictly increasing order. 
\item
Then the relation defined by $\varphi$ is given by
$$\{ (a_1,\dots , a_k)\in A^k\, |\, \mathfrak{A}\models \varphi(a_1,\dots , a_k)\, \}.$$
\end{enumerate}
Therefore, the relations defined by the (strictly speaking non-equivalent)
formulae $\varphi(x_1,\dots , x_k)$ and $\varphi(x_2,\dots , x_{k+1})$
will be the same. Note that the relation defined by a \emph{sentence} $\varphi$
such that $\mathfrak{A}\models\varphi$ is
the nullary relation $\{\emptyset\}$
where $\emptyset$ represents the empty tuple.
The relation defined by a sentence $\chi$
such that $\mathfrak{A}\not\models\chi$ is the 
nullary empty relation.
We suppose there is a \emph{different empty relation for each arity}, starting 
with arity zero. This way the complement of the completement of
the total $k$-ary relation is the total $k$-ary relation itself.
We lose no information about the arity.
We also assume that models must have a non-empty domain.

\medskip

\medskip

\noindent
We will next define functions that map relations in $\mathfrak{A}$
to relations in $\mathfrak{A}$. We will then show that this approach defines exactly 
the same relations as first-order logic.

\subsubsection{Functions on relations}\label{algebra}

Consider the algebraic signature $(u,I,\neg, p,s,\exists,J)$ where
\begin{enumerate}
\item
$u$ is a
nullary symbol,\footnote{Recall that nullary function symbols in an
algebraic signature represent constants.}
\item
$I,\neg,p,s,\exists$
have arity one,
\item
$J$ has arity two.
\end{enumerate}
To consider models with relation
symbols $R_1,\dots, R_k$, add $R_1,\dots, R_k$ to be nullary symbols in
the algebraic signature, just like $u$. Terms are
built from variable symbols and the constants (nullary symbols $u,R_1,\dots , R_k$)
using the symbols $I,\neg, p,s,\exists,J$ in
the usual way to compose new terms. Below we will consider only terms without
variable symbols and use the word ``term'' to refer to such terms.

\medskip

\medskip

\noindent
Given a model $\mathfrak{A}$, every term $\mathcal{T}$ defines
some relation $\mathcal{T}^{\mathfrak{A}} \subseteq A^k$ where $A$ is 
the domain of $\mathfrak{A}$. Let us look at the semantics of terms.
Let $\mathcal{T}$ be a term and suppose we have defined a relation $\mathcal{T}^{\mathfrak{A}}$.
Then the following conditions hold.

\medskip

\medskip

\noindent
\begin{enumerate}
\item[$R_i$\, )] Here $R_i$ is a $k$-ary relation symbol in the signature of $\mathfrak{A}$, 
which is a nullary term in the algebraic signature. The nullary term is
interpreted to be the relation $R^{\mathfrak{A}}$, i.e., the relation 
$$\{ (a_1,\dots, a_k)\, |\, \mathfrak{A}\models R(a_1,\dots, a_k)\, \}.$$
This is natural indeed.\footnote{Note here that if $R_i$ is a nullary relation symbol, $R^{\mathfrak{A}}$ is
either $\{\emptyset \}$ or $\emptyset_0$ corresponding to true and false, respectively.
Here $\emptyset_0$ is the nullary 
empty relation. The empty tuple is identified with $\emptyset$ in $\{\emptyset\}$.}
\item[$u$\, )] We define $u^{\mathfrak{A}} = A$. The constant $u$ is
referred to as the \emph{universal unary relation} constant.
\item[$I$\, )]
If $\mathcal{T}^\mathfrak{A}$ is of
arity at least two, we define
$$(\mathit{I}(\mathcal{T}))^{\mathfrak{A}}
= \{(a_1, \dots , a_k)\, |\, (a_1,\dots , a_k)\in\mathcal{T}^\mathfrak{A}
\text{ and }a_1=a_2\}.$$
If $\mathcal{T}^\mathfrak{A}$ is a
unary or a nullary relation, we define $(\mathit{I}(\mathcal{T}))^\mathfrak{A}
= \mathcal{T}^{\mathfrak{A}}$. The function $I$ is referred to as the \emph{identity operator}, or
equality operator.
\item[$\neg$\, )]
We define $$(\neg(\mathcal{T}))^{\mathfrak{A}}
= \{(a_1, \dots , a_{k})\, |\, (a_1,\dots , a_k)
\in A^k\setminus \mathcal{T}^\mathfrak{A}\, \},$$
where we recall that $A^0 = \{\emptyset \}$ in the
case where $k$ is zero.\footnote{When $k=0$,
the tuple $(a_1,\dots , a_k)$ represents the empty tuple $\emptyset$.}
The operator $\neg$ is referred to as the 
\emph{negation} operator or \emph{complementation} operator.
Recall that the empty relation is different for each arity.
\item[$p$\, )]
If $\mathcal{T}^\mathfrak{A}$ is of arity at least two, we define
\noindent
$$(\mathit{p}(\mathcal{T}))^{\mathfrak{A}}
= \{(a_2 , \dots , a_{k},a_1)\, |\, (a_1,\dots , a_k)\in\mathcal{T}^\mathfrak{A}\, \},$$
where the $k$-tuple $(a_2 , \dots , a_{k},a_1)$ is the one
obtained from the $k$-tuple $(a_1,\dots , a_k)$ by simply moving the first element $a_1$ to
the end of the tuple.
If $\mathcal{T}^\mathfrak{A}$ is a
unary or a nullary relation, we define $(\mathit{p}(\mathcal{T}))^\mathfrak{A}
= \mathcal{T}^{\mathfrak{A}}$. The function $p$ is referred to as the $\emph{permutation}$ 
operator, or $\emph{cyclic permutation}$ operator.
\item[$s$\, )]
If $\mathcal{T}^\mathfrak{A}$ is of arity at least two, we define
\noindent
$$(\mathit{s}(\mathcal{T}))^{\mathfrak{A}}
= \{(a_2, a_1, a_3,\dots , a_{k})\, |\, (a_1,\dots , a_k)\in\mathcal{T}^\mathfrak{A}\, \},$$
where the $k$-tuple $(a_2, a_1, a_3,\dots , a_{k})$ is the one
obtained from the $k$-tuple $(a_1,\dots , a_k)$ by swapping the 
first two elements $a_1$ and $a_2$ and keeping the other elements as they are.
If $\mathcal{T}^\mathfrak{A}$ is a
unary or a nullary relation, we define $(\mathit{s}(\mathcal{T}))^\mathfrak{A}
= \mathcal{T}^{\mathfrak{A}}$. The function $s$ is referred to as the \emph{swap} operator.
\item[$\exists$\, )]
If $\mathcal{T}^{\mathfrak{A}}$ has arity at least one, we define
$$(\exists(\mathcal{T}))^{\mathfrak{A}}
= \{(a_2, \dots , a_{k})\, |\, (a_1,\dots , a_k)\in\mathcal{T}^\mathfrak{A}
\text{ for some }a_1\in A\, \},$$
where $(a_2,\dots , a_k)$ is the $(k-1)$-tuple
obtained by removing the first element of the tuple $(a_1,\dots , a_k)$.
When $\mathcal{T}^{\mathfrak{A}}$ is a
nullary relation, we define $(\exists(\mathcal{T}))^{\mathfrak{A}} = \mathcal{T}^{\mathfrak{A}}$.
The function $\exists$ is referred to as the \emph{existence} operator.
\item[$J$\, )]
We define
\begin{multline*}
(J(\mathcal{T},\mathcal{S}))^{\mathfrak{A}} =\\
 \{(a_1,\dots , a_k, b_1 , \dots , b_{\ell})\, |\, (a_1,\dots , a_k)
\in \mathcal{T}^\mathfrak{A}
\text{ and }\, (b_1,\dots , b_{\ell})
\in \mathcal{S}^\mathfrak{A}\}.
\end{multline*}
Here we note that if $k=0$ and thus $(a_1,\dots , a_k) = \emptyset$ (the empty tuple),
then $(a_1,\dots , a_k, b_1 , \dots , b_{\ell})$
represents the tuple $(b_1 , \dots , b_{\ell})$.
Similarly, if $\ell = 0$, then $(b_1,\dots , b_{\ell}) = \emptyset$
and $(a_1,\dots , a_k, b_1 , \dots , b_{\ell})$ represents $(a_1,\dots , a_k)$.
When both $k$ and $\ell$ are zero, $(a_1,\dots , a_k, b_1 , \dots , b_{\ell})$
represents the empty tuple $\emptyset$. The function $J$ is referred to as the \emph{join} operator.
\end{enumerate}

\medskip

\medskip

\noindent
The terms that can be formed using the above 
symbols will be called \emph{$l$-terms} ($l$ for logic).
If $\varphi$ and an $l$-term define exactly the same
relation over every model $\mathfrak{A}$ (in a signature
interpreting the required symbols), then $\varphi$ and the $l$-term
are called \emph{$l$-equivalent}.

\medskip

\medskip

\noindent
The following theorem bears some similarity to Codd's theorem. However, we discuss
standard first-order logic and have a somewhat different 
operator set (and we concentrate on relations rather than
sets of assignments). Our signature is finite (provided 
that there are only finitely many relation symbols $R_i$ in the signature of
the models $\mathfrak{A}$ considered).

\medskip

\medskip

\begin{theorem}
For every first-order formula $\varphi$, there exists an $l$-equivalent $l$-term.
Vice versa, for every $l$-term, there exists an $l$-equivalent first-order formula.
\end{theorem}
\begin{proof}
Let $\varphi$ be a first-order formula. We need to find the 
corresponding algebraic term. If $\varphi$ is $\top$, the corresponding term is $\exists u$,
and if $\varphi$ is $\bot$, the term is $\neg \exists u$. If $\varphi$ is 
some formula $x=x$, then the term $u$ will do. If $\varphi$ is a formula $x = y$,
then the corresponding term is $I(J(u,u))$.

\medskip

\medskip

\noindent
Now suppose $\varphi$ is $R(x_{i_1},\dots , x_{i_k})$, where $k\geq 0$. Assume 
first that no variable symbol in the tuple $(x_{i_1},\dots , x_{i_k})$ gets repeated.\footnote{Thus for
example $R(x_1,x_1,x_2)$ is not allowed in this case as it repeats $x_1$.} Assume
also that the (subindices of the)
variable symbols in $(x_{i_1},\dots , x_{i_k})$ are linearly ordered (i.e.,
strictly increasing) from left to right. Then the
term corresponding to $\varphi$ is simply $R$.

\medskip

\medskip

\noindent
Consider then the cases where the
tuple $(x_{i_1},\dots , x_{i_k})$ in $R(x_{i_1},\dots , x_{i_k})$ may
contain repetitions and the variables may not
necessarily be in an increasing order. Firstly, note that we can permute any relation arbitrarily by
using the operators $p$ and $s$. To see this, 
note the following two facts.

\begin{enumerate}
\item
In a tuple $(a_1,\dots, a_i, \dots , a_{\ell})$, we can move the element $a_i$ any
number $n$ of steps to the right---keeping the tuple otherwise in the same order---as follows.
\begin{enumerate}
\item
Apply $p$ repeatedly so that $a_i$ becomes the leftmost element.
\item
Apply the composed function $ps$ (meaning ``$s$ 
first and then $p$'') exactly $n$ times.
\item
Repeatedly apply $p$ until the tuple is in
the desired configuration.
\end{enumerate}
\item
Moving $a_i$ to the left is no different from moving it to the right, as moving to the 
left corresponds to moving to the right and past the end of the tuple. Thus 
moving $n$ steps to the left is achieved by the above steps a,b,c, with
the combined function $ps$ applied exactly $\ell - n - 1$ times in step b.
\end{enumerate}

\noindent
Thus we can move a single element anywhere, keeping the rest of the tuple in order.
Thereby we can, one by one, move elements where we like, and thus all permutations can indeed be
achieved using $s$ and $p$ only.

\medskip

\medskip

\noindent
Notice then that since we can permute relations arbitrarily, also repetitions of variables can be
dealt with. The idea is simply to bring element pairs to the left end of tuples, after which we
can use the indentity operator $I$ to
get rid of tuples without the desired repetition. For example, it is easy to see that $R(x_2,x_1,x_2)$  
corresponds to the term $p\, \exists\, Ip\, p(R)$. It is easy to see how to systematically 
produce translations of all atoms by using 
combinations of $p$, $s$, $\exists$ and $I$.

\medskip

\medskip

\noindent
To translate a conjunction, suppose by induction that we have
translations $\mathcal{T}(\psi)$ and $\mathcal{T}(\chi)$ for $\psi$
and $\chi$. Let $\psi\wedge\chi$ be the formula $\varphi$ to be translated.
Now, $J(\mathcal{T}(\psi),\mathcal{T}(\chi))$ is almost what we need. The only 
thing we need additionally to take into account is the possibility of having repeated 
symbols that occur in both $\psi$ and $\chi$ and also the ultimate order of the 
variable symbols. Thus, similarly to the case for
atoms, we apply $p$, $s$, $\exists$ and $I$ (often repeatedly) to the
term $J(\mathcal{T}(\psi),\mathcal{T}(\chi))$ to get the required term.

\medskip

\medskip

\noindent
Translating a negation is trivial, we translate $\neg\psi$ to
the term $\neg(\mathcal{T}(\psi))$ where $\mathcal{T}(\psi)$ is obtained from the induction hypothesis.
Translating a quantifier $\exists x_i$ is similarly easy. However, we may first have to do some 
preprocessing  as the variable $x_i$ can refer to some other than the first position in the relation 
corresponding to the quantified formula. Thus,
suppose we want to translate $\exists x_i\psi$ and we have a
translation $\mathcal{T}(\psi)$ of $\psi$ by the induction hypothesis. Now use $p$ (typically
repeatedly) to make
the coordinate corresponding to $x_i$ the leftmost coordinate in the relation tuples, obtaining a
term $p^n(\mathcal{T}(\psi))$, where $n$ denotes how many times $p$ was repeated.
Then use $\exists$ and use $p$ again (typically repeatedly) to put the remaining tuple into
the right order. Thus the ultimate term is of type $p^m\exists\, p^n(\mathcal{T}(\psi))$.
%
%
%
%
%
%

\medskip

\medskip

\noindent
The direction from terms to first-order logic is straightforward.
\end{proof}

\medskip

\medskip

\noindent
This representation of first-order logic can be used to obtain very fine-grained 
classifications of first-order fragments. Thus it can be a fruitful starting point for
novel classifications of different decidability and complexity issues 
of first-order fragments.\footnote{Of course complexities
can vary based on which formalism is used. One can use the algebraic formalism for
classifying first-order fragments for sure, but one can also directly use and study the
algebraic formalism itself
and complexity issues within it.}
For example, it seems plausible to expect that some fragments with fluted-logic-like
properties can be obtained via dropping the swap
operator $s$. Anyway, there are many ways to directly apply the framework, and it
should be a nice and useful setting for building
decidability and complexity classifications based on simply and clearly specified  fine-grained
classifications of syntax.

\medskip

\medskip

\noindent
The algebraic approach generalizes to second-order logic quite directly. There we 
can make use of relations whose tuples have individuals and relations. We leave this for later.
Instead, let us define a general operator within this framework. It is now particularly easy as the 
algebraic framework is indeed very simple.
So, let $F$ be a map that takes as input any tuple $(M,R_1,\dots , R_k)$ where $M$ is a set
and each $R_i$ a
relation over $M$. Note that the collection of inputs is indeed all
tuples $(M,R_1,\dots , R_k)$ so the relations $R_i$ can have
any sequence of arities for different inputs (but $k$ is fixed).
The output $F((M,R_1,\dots , R_k))$ is then some relation over $M$.
The invariance condition is simply that if $(M,R_1,\dots , R_k)$ and $(M',R_1',\dots , R_k')$
are isomorphic\footnote{We note that there is no vocabulary here, as the 
relations here are indeed plain relations not directly associated with any
relation symbol.} via $f: M \rightarrow M'$,
then so are $(M,F((M,R_1,\dots , R_k)))$ and $(M,F((M',R_1',\dots , R_k')))$ (also via $f$).
We now add terms $F(t_1,\dots , t_k)$ to the picture, and the semantics 
obviously defines $(F(t_1,\dots , t_k))^\mathfrak{M}$ to
be $F(t_1^{\mathfrak{M}},\dots , t_k^{\mathfrak{M}})$.

\bibliographystyle{plain}
\bibliography{aaabibfile}

\end{document}